\documentclass{ectj}
\usepackage{natbib}
\usepackage{silence}
\usepackage{amsfonts,amssymb,graphics,epsfig,verbatim,bm,latexsym,amsmath,url,amsbsy}
\usepackage{bbm}
\usepackage{bm}
\usepackage{array}
\usepackage{blkarray}
\usepackage[style=base]{caption}
\usepackage[ruled,vlined]{algorithm2e}
\usepackage{pst-all}
\usepackage{tabularx,booktabs}
\usepackage{multirow}
\usepackage[retainorgcmds]{IEEEtrantools}
\usepackage[margin=3cm]{geometry}

 \usepackage{setspace}
\newtheorem{theorem}{Theorem}

\newtheorem{proposition}{Proposition}

\newtheorem{definition}{Definition}
\definecolor{Gray}{gray}{0.9}

\usepackage{mathtools}

\DeclarePairedDelimiter\floor{\lfloor}{\rfloor}

\renewcommand{\thesection}{\arabic{section}}
\renewcommand{\theequation}{\arabic{section}.\arabic{equation}}

\newcounter{bean}

\received{...}
\accepted{...}
\volume{21}

\title[Learning Nesting Structure]{Learning Structure in Nested Logit Models}

\author[Aboutaleb et al.]{Youssef M. Aboutaleb$^{\dagger}$, Moshe E. Ben-Akiva $^{\dagger}$, Patrick Jaillet $^{\dagger}$}

\address{$^{\dagger}$ Massachusetts Institute of Technology,
            77 Massachusetts Avenue, Cambridge, MA 02139, USA.}
\email{ymedhat@mit.edu, mba@mit.edu, jaillet@mit.edu}

\def\AmSTeX{$\cal A$\kern-.1667em\lower.5ex\hbox{$\cal M$}\kern-.125em
    $\cal S$-\TeX}
\def\BibTeX{{\rm B\kern-.05em{\sc i\kern-.025em b}\kern-.08em
    T\kern-.1667em\lower.7ex\hbox{E}\kern-.125emX}}

\begin{document}
  \begin{abstract}
This paper introduces a new data-driven methodology for nested logit structure discovery. Nested logit models allow the modeling of positive correlations between the error terms of the utility specifications of the different alternatives in a discrete choice scenario through the specification of a nesting structure. Current nested logit model estimation practices require an \textit{a priori} specification of a nesting structure by the modeler. In this we work we optimize over \textit{all} possible specifications of the nested logit model that are consistent with rational utility maximization. We formulate the problem of learning an optimal nesting structure from the data as a mixed integer nonlinear programming (MINLP) optimization problem and solve it using a variant of the linear outer approximation algorithm.  We exploit the tree structure of the problem and utilize the latest advances in integer optimization to bring practical tractability to the optimization problem we introduce. We demonstrate the ability of our algorithm to correctly recover the true nesting structure from synthetic data in a Monte Carlo experiment. In an empirical illustration using a stated preference survey on modes of transportation in the U.S. state of Massachusetts, we use our algorithm to obtain an optimal nesting tree representing the correlations between the unobserved effects of the different travel mode choices. We provide our implementation as a customizable and open-source code base written in the Julia programming language. 
\keywords{Nested Logit, Discrete Choice, Algorithmic Model Selection, Machine Learning}

    \end{abstract}


 \section{Introduction}
 \setcounter{equation}{0}

  
 First introduced by \cite{benakiva-1972} in the context of travel demand modeling and extended by \cite{mcfadden-1978} in the context of resdiential location choice, the nested logit model is widely used in the marketing and transportation science literatures among others for modeling the choice a rational decision-making agent makes from a set of mutually exclusive alternatives (\citealp{benakiva-1985}). That choice might be the travel mode of daily commute, a flight to book, or an economics textbook to buy. The nested logit model belongs to a wider class of behavioral choice models known as random utility models (RUMs).
 
RUMs turn on the assumption that the decision maker ranks the available choices (or alternatives) in order of preference as represented by a latent (i.e., unobserved from the modeler's perspective) utility function. The utility that a decision-making agent associates with a particular alternative is assumed to be comprised of a systematic component and an error term. The systematic component is specified as a function (typically linear-in-parameters) of the attributes of the alternative and characteristics of the decision maker. The error term is modeled as a random variable. There is a utility equation for each alternative. Logit models are a family of RUMs that assume a joint extreme-value distribution of the error terms across the alternatives. Further assumptions on the the error terms lead to different types of logit models and give rise to different mathematical forms for the choice probabilities.
 
 The multinomial (flat) logit model assumes independence and homoscedasticty of the error terms across the different alternatives. The independence assumption implies that the odds ratios between any two alternatives are independent of the attributes or availability of the other alternatives. This property, called the independence from irrelevant alternatives (IIA), is convenient from an estimation stand point but is an unrealistic restriction on agent behavior in some applications. One implication of the IIA assumption is ``proportional substitution" between the alternatives. The red bus/blue bus paradox where the introduction of a blue bus travel mode, with the same attributes as an existing red bus travel mode, draws equally from the market shares of \textit{all }of the available alternatives (instead of splitting the red bus market share), is a classic example of a choice setting where the IIA property is inappropriate. Demand predictions can be seriously compromised by incorrectly assuming the IIA property (\citealp{mcfadden1977application}).

 The nested logit model is invoked when it is thought that there are common unobserved attributes of the choice alternatives. The independence assumption is relaxed by partitioning the alternatives into mutually exclusive and collectively exhaustive subsets called “nests”. The error terms in the utility expressions are decomposed into alternative specific and nest specific error terms, the latter of which introduces correlations between the collective error terms of alternatives sharing a common nest. This breaks the independence assumption of the multinomial logit and allows the modeling of more general substitution patterns. The homoscedasticty assumption remains, the total variance of the error terms is assumed to be the same for all the alternatives. Each nest is associated with a scale parameter which quantifies the variance of the nest specific error terms and consequently the covariance between the collective error terms of any two alternatives sharing the nest. 

There are references in the literature to two ``different" nested logit models: a non-normalized nested logit (NNNL) and a utility maximization nested logit (UMNL) model (\citealp{koppelman1998alternative}; \citealp{koppelman1998nested} ). These two models where shown to be equivalent-- the difference being in the normalization of the scale parameters; see \cite{hensher2002specification}.

To specify the nested logit model, a nesting partition of the alternatives is necessary. In some applications, the appropriate partition to use is immediately clear as in \cite{mcfadden-1978} where residential location choice is made first by community and then by dwelling type. In many other applications, however, the partition choice is \textit{ad hoc}. That the estimation results (including the policy parameters in the systematic specification) are affected by the partition choice is viewed as a problematic aspect of the nested logit model; see \cite{greene2003econometric}.

The large set of possible partitions precludes an exhaustive search and presents a significant modeling challenge in deciding which partitioning of the alternatives best reflects the underlying choice behaviour of the population. Typically, shortcuts are made on the basis of \textit{a priori} reasoning in determining a partition. The current \textit{modus operandi} is to rely on domain knowledge to substantially reduce the feasible set to a small set of candidate partitions and employ statistical techniques to determine goodness of fit. Many studies in literature present several sets of results based on different partitions. This is done at the risk of potentially excluding some ostensibly non-intuitive structures which might actually provide a better description of the choice behaviour of the population under study; see \cite{koppelman-2006} and \cite{andrew-tree}. This motivates the approach we take in this paper for a more \textit{holistic} view of nested logit model estimation, i.e., one that optimizes over the nesting partition as well as model parameters.
 
The desire for finding an optimal nesting tree (or partition) was first articulated in the literature by \cite{andrew-tree} who spoke of an ``ultimate need" of a method to empirically identify an optimal structure. \cite{andrew-tree} formulated the likelihood as a function of the nesting tree and the model parameters but remarked that since the function is discrete the only reliable method (at the time) was to evaluate the maximum of the likelihood at different candidate tree specifications and choose the specification that provides the best fit. A key intuition from \cite{andrew-tree}'s formulation is that the optimization has to be made \textit{jointly} with nesting structure, scale parameters, and model parameters in the systematic specification. Indeed the optimal structure is \textit{not} independent of the systematic specification. We indulge on this point later in this section. The desire for a systematic method for learning nested logit structure has been reiterated in many works including \cite{bierlaire-book} and \cite{greene2003econometric}.

Testing departures from the IIA property is sometimes used to guide the specification of the nesting partition. Hausman’s specification test (\citealp{hausman1984specification}) turn on the observation that if a subset of the alternatives is truly irrelevant then its omission, although might lead to inefficiencies, will not systematically affect the estimated parameters. The test compares the estimated parameters with and without the subset of alternatives to determine is the IIA is violated. Hausman’s test of course requires the researcher to specify which subsets to omit and therefore can not be the basis of a practical nested structure learning algorithm since there is exponentially many such subsets of alternatives to test. 

Perhaps the first attempt at automating the nested logit structure specification appears in \cite{benson2016relevance}. The authors notice that an alternative $k$ is an irrelevant alternative of alternatives $i$ and $j$ if and only if including it in the choice set does not significantly affect the pairwise preference of item $i$ over item $j$. They describe a battery of statistical tests to compare the probability of choosing $i$ over $j$ and formulate an algorithm to systematically test the IIA assumption and build a corresponding nesting tree. Once the structure is discovered in this way, the authors then specify a systematic component for the choice model and estimate its the parameters. 

The method of \cite{benson2016relevance} fails to account for the interdependence between the optimal nesting structure and the systematic specification of the utility equations. Recall that the error terms are defined simply as the difference between the true latent utility and the systematic utility specified by the modeler, therefore the independence (or lack thereof) of the error terms depends on the specification of the systematic utility. For a given choice situation, two different specifications of systematic utility will result in two different sets of error terms. One set might be independent while the other might not. Therefore, the IIA assumption might hold for one specification of systematic utility but not for another, even though both specifications relate to the same choice situation. This means that, strictly speaking, the IIA property is or is not valid for a particular specification of systematic utility in a logit model of a particular choice situation, not for the choice situation itself; see \cite{mcfadden1977application} for a more complete discussion.

In this work, we take the full likelihood approach first suggested by \cite{andrew-tree}. We introduce and formulate the nested logit structure learning problem (NLSLP) as a linearly constrained mixed integer nonlinear programming (MINLP) problem. We introduce a solution algorithm based on the linear outer approximation algorithm (\citealp{fletcher-1994}).  Over the past 25 years, algorithmic advances in integer optimization coupled with hardware improvements have resulted in an 800 billion factor speedup in mixed-integer optimization (\citealp{bertsimas2017optimal}). Furthermore, besides the massive speed increases, the introduction of ``lazy constraints" implementation in modern-day solvers (which provide for an efficient way of dealing with constraints that grow exponentially in number with the size of the problem) has brought tractability to many problems of practical interest (\citealp{pearce2019towards}). This has enabled recent successes when applying modern mixed integer optimization methods to a selection of statistical problems (\cite{bertsimas2016best}; \cite{bertsimas2014least}; \cite{bertsimas2017logistic}; \cite{bertsimas2016or}; \cite{aboutaleb2020sparse}). We make use of the state-of-the art in algorithmic development in solving mixed integer optimization problems and lazy constraints to bring practical tractability to the NLSLP.

 The remainder of this paper is organized as follows. Section 2 introduces the notations we use throughout the paper and reviews the relationship between nesting trees and the variance-covariance matrix of the error terms in a nested logit model. We also introduce an important result that brings needed tractability to the NLSLP. In Section 3, we formally introduce the NLSLP and describe in detail its formulation. In Section 4, we present a complete solution algorithm for the NLSLP. Section 5 presents Monte Carlo simulations to validate our proposed methodology, in addition to an empirical application. Section 6 concludes the paper. 
  \section{Technical Setup and Foundational Results}
 \setcounter{equation}{0}
 In this section we introduce the notation we use throughout the paper and present a result that brings tractability to the nested logit structure learning problem that we use in our formulation Section 3. 
 In accordance with the notation in \cite{benakiva-1985}, let $U_{in}=V_{in}+\epsilon_{in}$ be the utility agent $n\in \mathcal{I}$ associates with alternative $i\in \mathcal{C}_n \subseteq \mathcal{C}$. $V_{in}$ is the systematic utility component (a deterministic, typically linear, function of model ``taste" parameters, attributes of the alternatives, and characteristics of the agent) and $\epsilon_{in}$ are extreme-value distributed errors. The error terms represent the effects of omitted taste variations, choice attributes and socioeconomic variables. 
 Under the nested logit framework, the universal choice set $\mathcal{C}$ is partitioned into mutually exclusive and collectively exhaustive ``nests" formally defined as follows:

\begin{definition}
A \textit{\textbf{nested partition}}, $\mathcal{B}$, of a set $\mathcal{C}$ is a set of nonempty subsets $B_m \subseteq \mathcal{C}$ such that (1) One of the subsets $B_m$ is the universal set $\mathcal{C}$. (2) The subsets are nested, i.e., whenever two distinct subsets are overlapping ($B_m \cap B_m' \neq \emptyset$), one subset contains the other (either $B_m \subseteq B_m'$ or $B_m' \subseteq B_m$).
The subsets $B_m$ constituting a nesting partition $\mathcal{B}$ as referred to as ``nests".
\end{definition}

A nest $B_m \in \mathcal{B}$ is called \textit{degenerate} if $|B_m|=1$. A nested partition $\mathcal{B}$ is called degenerate if it contains one or more degenerate nests.
 
\begin{definition}
Given a nested partition $\mathcal{B}$, let $B(j)$ denote the smallest nest $B_m \in \mathcal{B}$ containing the elemental alternative $j \in \mathcal{C}$. 
Similarly, $B(\mathcal{C})$ denotes the smallest subset $B_m \in \mathcal{B}$ such that $\mathcal{C} \subseteq B_m$.
\end{definition}

Having introduced nested partitions in terms of subsets, we now introduce the more familiar graph based representation of a nested partition i.e., nesting trees. The machinery of graph theory will be used in the next section to enforce the desired \textit{arborescence} (tree) properties in the optimization formulation. The subset-based representation is convenient from a notation standpoint, and is useful when proving certain properties such as counting the total number of possible partitions. 

The following theorem states that if we don't allow degenerate nests, the maximum number of possible nests in a nesting partition is two less than the number of elemental alternatives. The proof of the theorem also shows the process of building a tree representation of a nested partition. 

\begin{theorem}
Let $\mathcal{C}$ be a finite set with $n$ elements, where $n\geq 2$. Any non-degenerate nested partition of $\mathcal{C}$ can be represented by a unique tree with $n$ leaves and at most $n-2$ internal nodes.
\end{theorem}

\begin{proof}
Suppose $\mathcal{B}$ is a non-degenerate nested partition of $\mathcal{C}$. To construct a tree representation, $\mathbb{T}$, of $\mathcal{B}$: 
\begin{list}
{\textsc{Step} \arabic{bean}.}{\usecounter{bean}}   
\item Create a \textit{root} node ($\circ$) and let it represent the universal choice set $\mathcal{C}$.
    \item Represent each nest $B_m \in \mathcal{B}$ containing a strict subset of the alternatives in $\mathcal{B}$ by an \textit{internal} node ($\bigtriangleup$).
        \item Represent each elemental alternative $j \in \mathcal{C}$ by a \textit{leaf} node ($\bullet$).
    \item Connect each leaf node to the internal node corresponding to the smallest nest containing it. Formally for each $j \in \mathcal{C}$ find $B(j)$ and draw a directed edge from the node representing $B(j)$ to the leaf node representing $j$.
    \item Similarly, connect each nest to the smallest set containing it. For each $B_m \in \mathcal{B}$ find $B(B_m)$ and connect the node representing $B(B_m)$ to that representing $B_m$.

\end{list}
Property 1 of Definition 2.1 guarantees that each $B(.)$ is defined, and property 2 guarantees that each $B(.)$ is unique (since degeneracy is not allowed).

We now show that the number of nest nodes is at most $n-2$. Let $b$ denote the number of nest nodes. The $n$ leaf nodes are terminal nodes and must have degree one. Since the nesting partition is non-degenerate by assumption, all internal nodes must have at least two children. This implies that (i) the root has at least degree two, and (ii) the nest nodes have at least degree three (one parent and at least two children). Therefore $deg(\mathbb{T}) \geq n +3b+2$. Now, by the \textit{Handshaking lemma} (\citealp{vasudev2006graph}) we have that the degree of the graph is twice the number of edges: $deg(\mathbb{T}) = 2|\mathcal{E}|$. Furthermore, since $\mathbb{T}$ is a tree, the number of edges is one less than the number of nodes so that $|\mathcal{E}|=(n+b+1)-1$, therefore, $2(n+b) \geq n + 3b+2$ and hence $b \leq n-2$.
\end{proof}\hfill $\square$\\
As an illustration, Figure 1 shows the four possible nesting partitions and their corresponding tree representations for the universal choice set $\mathcal{C}=\{1,2,3\}$.

Theorem 2.1 brings tractability to the problem of learning nested logit structures. Namely, in an optimization framework, we can simply include the maximum number of nests allowed and let an optimization procedure guide the inclusion or exclusion of these nests. Degenerate nests do not affect the likelihood and can be omitted without loss of generality.

The nesting structure determines \textit{how} the alternatives are correlated, and the associated scale parameters determine by \textit{what amount} they are correlated. The exact relation between the nesting structure and correlation of the error terms is laid out in the following proposition; see \cite{daganzo1993two} and \cite{galichon2019representation} for a theoretical derivation.
\begin{proposition}
Let $\mathcal{B}$ be a nested partition of the choice set $\mathcal{C}$. If the nested logit assumptions are satisfied, then the covariance $cov(\epsilon_i,\epsilon_j)$ between the error terms in the utility expression of any two distinct alternatives $i,j \in \mathcal{C}$ is given by:
\begin{equation*}
    \frac{\pi^2}{6}\big(\frac{1}{\mu_{r}^2}-\frac{1}{\mu_{B(\{i,j\})}^2}\big)
\end{equation*}
where $\mu_{r}$ is the scale of the root node and $\mu_{B(\{i,j\})}$ is the scale of the nest node corresponding to the smallest partition containing both $i$ and $j$.
\end{proposition}
Since the overall scale of the utility equations is not defined, only $|\mathcal{C}|-1$ scale parameters may be identified. We follow the convention of normalizing the scale of the root nest, $\mu_r$, to 1. The nested logit model is consistent with rational utility maximization only if the scale parameters increase with increasing nesting levels; see \cite{borsch1990compatibility}. Formally the scale parameters must satisfy:
\begin{equation}
    B_m \subseteq B_{m'} \implies \mu_{m'} \leq \mu_m
\end{equation}
Note that this constraint on the scale parameters restricts that covariances to be  non-negative (this should be clear from the expression in Proposition 2.1). Intuitively, the covariance between two error terms in a nested logit model is the variance of the nest specific error term of the shared nest which can not be negative; see \cite{williams1982behavioural} for details and \cite{dong2017negative} for modeling implications.
Furthermore by the same proposition, two alternatives are uncorrelated if the smallest partition containing both is the universal set $\mathcal{C}$ (or equivalently their youngest common ancestor is the root node). If a nested partition consists only of the universal set $\mathcal{C}$, the nested logit model reduces to the multinomial logit model. Figure 2 illustrates such a partition, and the corresponding tree structure (c.f. Theorem 2.1) and covariance matrix (c.f. Proposition 2.1). Figure 3 shows another example for the same choice set, but where the nested partition is given by $\{\{a_1,a_2,a_3,a_4\},\{a_2,a_3,a_4\},\{a_3,a_4\}\}$.

Given a nested partition $\mathcal{B}$ over $\mathcal{C}$, consider the probability $P_n\{j\}$ of agent $n\in \mathcal{I}$ choosing alternative $j \in \mathcal{C}$. Using the product rule and conditioning over all the nests in $\mathcal{B}$ containing $j$ in ascending order of cardinality, we have:
\begin{align}
    P_n\{j\}=P_n\{j|B(j)\}P_n\{B(j)|B^2(j)\} \ldots P_n\{B^{k-1}(j)|B^k(j)\}P_n\{B^k(j)\}
\end{align}
Where $B^s(.)$ is the s-fold application of the function $B(.)$, and $k$ is the number of partitions in $\mathcal{B}$ that contain the alternative $j$. Note that $B^k(j)=\mathcal{C}$, i.e., the largest set containing the alternative $j$ is the set $\mathcal{C}$. $P_n\{j|B(j)\}$ is the conditional probability of choosing $j$ given that some alternative in $B(j)$ is chosen. This quantity is given by the relative attractiveness of alternative $j$ compared to the maximum utility obtainable by choosing some alternative in $B(j)$, which we denote by $\Gamma_n\{B(j)\}$.
The Inclusive value, $\Gamma_n\{B_m\}$, of a nest $B_m$ is its expected maximum utility to agent $n$ (\citealp{benakiva-1985}), defined as follows:
\begin{equation}
    \Gamma_n\{B_m\} =\frac{1}{\mu_{B_m}} \ln{ \big(\sum_{j \in \mathcal{C}| B(j)=B_m} e^{\mu_{B_m}V_{jn}}} + \sum_{B_{m'} \in \mathcal{B}|B_{m'} \subseteq B_m} e^{\mu_{B_m}\Gamma_n\{B_{m'}\}}  \big)
\end{equation}
So that
\begin{align}
    P_n\{j|B(j)\}=\exp(\mu_{B(j)}(V_j-\Gamma_n\{B(j)\}))
\end{align}
Similarly, the probability of choosing nest $B^s(j)$ conditional on choosing $B^{s+1}(j)$, the smallest nest containing it, is given analogously to (2.3):
\begin{align}
    P_n\{B^s(j)|B^{s+1}(j)\}=\exp(\mu_{B^{s+1}(j)}(\Gamma_n\{B^{s}(j)\}-\Gamma_n\{B^{s+1}(j)\}))
\end{align}
Finally, since \textit{some }alternative in $\mathcal{C}$ has to be chosen, we have $P_n\{B^k(j)\}=1$.
Putting it all together, we have
\begin{align}
    P_n\{j\}=\exp\Big({\mu_{B(j)}V_{jn} +\sum_{i=1}^{k-1}(\mu_{B^{i+1}(j)}-\mu_{B^i(j)})\Gamma_n\{B^i(j)\}-\mu_{B^{k}(j)}\Gamma_n\{B^{k}(j)\}}\Big)
\end{align}
The closed-form (2.6), while mathematically tractable, is a rather complicated function involving nested logs of sums of exponential terms. We discuss the implications of this on the solution algorithm in Section 4.

In this section we reviewed the modeling flexibility the nested logit model brings in the specification of the variance-covariance matrix of the error terms while maintaining closed-form expressions for the choice probabilities. 
Our motivation is to algorithmically and systematically determine which of the allowable variance-covariance matrix
structures (as determined by the nesting trees) best fits the data.
 Working with the variance-covariance matrix directly however is not very attractive in the nested logit framework because of the many constraints on which matrix structures are permissible in that framework. We introduced notations and equivalences between three different representations of nested logit models namely 
 (i) a nested partition based representation that was used to formulate the choice probabilities, (ii) a tree based representation that was used to prove Theorem 2.1, and will be useful in the optimization framework we introduce in the next section, and (iii) a variance covariance based representation which helped us understand the flexibility and limitations of the nested logit model.
      \section{The Nested Logit Structure Learning Problem}
 \setcounter{equation}{0}
The general framework we take for finding an optimal nesting tree $\mathbb{T}$ (representing a nesting partition) and its parameters $\theta_{\mathbb{T}}$ (the scale and taste parameters) is a mixed-integer non-linear program. At a high level, the problem can be written as follows:
\begin{align}
    \mathbf{max_{\mathbb{T,\theta_T}}} & \textnormal{ Model Likelihood}\\
    \mathbf{subject \; to \;} & \textnormal{}\mathbb{T} \textnormal{ is a valid nesting tree}\\
    &\textnormal{Scale parameters are non-decreasing with nesting level }\\
    & \textnormal{Number of nests} = M \\
    & \textnormal{Nesting levels} = L 
\end{align}
The objective is to optimize over all possible nesting tree specifications. This is achieved by maximizing the model likelihood while penalizing model complexity to avoid over-fitting. The constraints restrict the search to valid nesting trees that satisfy the utility maximization constraints given by equation (2.1).

Constraint (3.2) guarantees that $\mathbb{T}$ is a valid nesting tree which requires the prohibition of cycles and enforcement of graph connectivity. Constraint (3.3) places restrictions on the scale parameters of the tree. In order for the estimated parameters to be consistent with utility maximization the scale parameters can not decrease with nesting level. The last two constraints, (3.4) and (3.5), limit the complexity of the model. The optimal parameters $M$ and $L$ are determined through cross-validation (by evaluating the likelihood on a hold out ``validation" data-set).

In this section, we obtain a closed form expression for the likelihood (Section 3.1) as a function of nesting structure and model parameters and mathematically formulate (3.2)-(3.5) as linear constraints (Section 3.2). We end with a discussion on the regularization framework we adopt in Section 3.3.

\subsection{A Mixed-Integer Nonlinear Program Formulation}

Let $\mathcal{B}$ be a non-degenerate nested partition of the set of alternatives $\mathcal{C}$ and consider the graph representation $\mathbb{T}$ of $\mathcal{B}$. $\mathbb{T}$ is a tree (cf. Theorem 2.1) and in building $\mathbb{T}$ from $\mathcal{B}$ we have: 
\begin{enumerate}
\item A root node $r$ representing the choice set $\mathcal{C}$.
\item Internal nest nodes $\mathcal{N}$ each representing a nest in $\mathcal{B} \setminus \mathcal{C}$.
\item Leaf nodes representing each alternative in $\mathcal{C}$.
\end{enumerate}

Formally, $\mathbb{T}=(\mathcal{V},\mathcal{E})$ is a directed graph where $\mathcal{V}=\{r\}\cup \mathcal{N} \cup \mathcal{C}$ is the set of vertices and $\mathcal{E}$ the set of edges. Let $x_{u,v}\in \{0,1\}$ equal one if there is a directed edge between nodes $u$ and $v$. There is a directed edge between two nodes $u$ and $v$ in $\mathbb{T}$ only if node $u$ is the nest or root node representing the smallest nesting partition containing the nest or alternative represented by $v$. Formally:
\begin{equation}
    x_{u,v}=1 \iff B(u)=v
\end{equation}

 In learning $\mathbb{T}$ from the data, we do not know \textit{a priori} the structure of the tree, or if nesting is present. The goal is to use optimization to reveal this structure. 
Theorem 2.1, provides the following guarantee: 
\begin{align}
    |\mathcal{N}|\leq |\mathcal{C}|-2
\end{align}
For convenience let $p=|\mathcal{C}|-2$. Now, since any nested partition can be represented by at most $p$ nest nodes, we start with the nest node set $\mathcal{N}$ containing that maximum number of nest nodes (i.e., $p$), and let an optimization procedure guide the inclusion or exclusion of these nests. To this end, we define for every nest node $v \in \mathcal{N}$, a binary decision variable $y_v \in \{0,1\}$ equal to one if nest $v$ is included in nesting tree $\mathbb{T}$. We shall collectively denote $\{x_e, \; e\in \mathcal{E}\}$ by $\mathbf{x}$. Similarly, we denote $\{y_v, \; v\in \mathcal{N}\}$ by $\mathbf{y}$.

\subsubsection{Objective Function}
In this section we obtain a closed form expression of the log-likelihood function. Let $c_{na}$ be a binary variable indicating if agent $n\in\mathcal{I}$ chose alternative $a\in\mathcal{C}$. Recall that $V_{an}$ is the systematic utility of alternative $a$ to agent $n$ and is specified a (linear) function of model parameters $\bm{\beta}$, attributes, and socioeconomic variables. The probability of choosing some alternative $j \in \mathcal{C}$ can be found by conditioning on the path from the root $r$ to the leaf node $j$. Let $\{r\xrightarrow{}j\}_{\mathbb{T}}$ denote the set of all possible paths from $r$ to $j \in \mathcal{C}$ on graph $\mathbb{T}$. If $\mathbb{T}$ is a tree, the path is unique by definition (\citealp{korte2012combinatorial}). Formally, $\{r\xrightarrow{}j\}_{\mathbb{T}}$ is a set of sets of ordered sequence of nodes visited on the path from $r$ to $j$. For any given path $l \in \{r\xrightarrow{}j\}_{\mathbb{T}}$ denote these nodes by $b^{(1)}_{l},...,b^{(s)}_{l}$ where $b^{(1)}_{l}=r$ and $b^{(s)}_{l}=j$, where $s$ is the length of the path $l$.
We can now write the log-likelihood as:
    \begin{align}
       \mathcal{L(}\mathbf{x,\bm{\beta},\bm{\mu}})&= \sum_{n\in\mathcal{I}} \sum_{a\in\mathcal{C}} \Big(c_{na} \sum_{l\in\{r\xrightarrow{}a\}}x_{l} \ln{{P}_n\{a|l\}}) \Big)     
    \end{align}

Where,\begin{align}
    \ln{{P}_n\{a|l\}}&= \mu_{b^{(s-1)}_{l}}V_{an} +\sum_{i=2}^{s-2}(\mu_{b^{(s-i)}_{l}}-\mu_{b^{(s-i+1)}_{l}})\Gamma_n\{b^{(s-i+1)}_{l}\}
    +(\mu_{r}-\mu_{b^{(2)}_{l}})\Gamma_n\{{b^{(2)}_{l}}\}-\mu_r\Gamma_n\{r\}
\end{align} and $\Gamma_n\{v\}$ is the inclusive value of internal node $v \in \mathcal{N} \cup \{r\}$ for agent $n$:
\begin{equation}
    \Gamma_n\{b\} =\frac{1}{\mu_{b}} \ln{ \big(\sum_{a \in \mathcal{C}} x_{ba}e^{\mu_{b}V_{an}}} + x_{bb'}\sum_{b' \in \mathcal{N}} e^{\mu_{b}\Gamma_n\{b'\}}  \big),
\end{equation}

The likelihood function (3.8) has an exponential number of terms and even for simple graphs can take exponential time to evaluate ``top-down" i.e., by enumerating all possible paths from the root to each alternative. However at valid tree solutions, there is a unique path from the root to each alternative. In such cases to compute the likelihood, for each individual $n$, one starts at the chosen alternative, $c_{na}$, and follows the alternative's ancestry adding to the utility of the alternative the inclusive values along the unique path to the root and scaling appropriately. We discuss matters regarding the efficient evaluation of the likelihood in Section 4.2.1.

\subsection{Constraints}
There are two main types of constraints: constraints that guarantee (i) a valid nesting tree and (ii) consistency with rational utility theory. There are also additional constraints on the structure of the tree $\mathbb{T}$. We discuss these individually in what follows. Crucially, we show that all the desired properties can be enforced using \textit{linear} constraints.
\subsubsection{Arborescence}
A tree with $m$ nodes must have exactly $m-1$ edges (otherwise the addition of an edge results in a cycle and the removal of an edge results in a disconnected graph). The following constraint guarantees that the total number of edges is one less than the total number of nodes.
\begin{align}
     \underbrace{\sum_{e \in \mathcal{{E}}}x_{e}}_\text{{Number of edges~}}=\underbrace{\big(\sum_{u \in \mathcal{N}}y_u + |\mathcal{C}| +1\big)}_\text{{Number of nodes~}} -1
\end{align}
To guarantee that the graph is cycle free, \textit{any} strict subset of the nodes of the tree must have at most one less edge than the number of nodes in said subset. The following of constraints takes any potential cycle and declares it illegal:
\begin{equation}
    \sum_{\{(u,v) \in \mathcal{E}: u\in \mathcal{A},  v\in \mathcal{A}\}} x_e \leq |\mathcal{A}| -1, \;\; \forall \mathcal{A}\subset \mathcal{V}
\end{equation}
This formulation is based on the sub-tour elimination constraints from the classic travelling salesman problem (\citealp{berstimas-book}). Since the number of subsets that can be formed from a given set is exponential in the size of the set, the number of cycle elimination constraints is exponential in the number of nodes of the tree. We discuss a very practical and efficient way of addressing this in Section 4.2.2 using ``lazy constraints".
\subsubsection{Scale Constraints}
In order for the estimated model to be consistent with utility maximization we require the scale parameters $\mu$ to increase with nesting level. The implication is that whenever $ x_{uv}=1$ we must have that $\mu_u\leq\mu_v$.
This can be enforced through the following constraints:
\begin{equation}
    \mu_u-\bar{\mu}(1-x_{uv})\leq \mu_v \; \; \forall \textnormal{ distinct } u,v \in \mathcal{N}
\end{equation}
where $\bar{\mu}$ is an upper bound on the scale parameters.

\subsubsection{Structure Constraints}
We describe a number of constraints on the structure of the graph. Let $\delta^{out}_u =\{(u,v)\in \mathcal{E} \}$ denote the set of edges that\textit{ originate} in node $u$, and similarly let $\delta^{in}_u =\{(v,u)\in \mathcal{E} \}$ denote the set of edges that \textit{terminate} in node $u$

Each leaf node must belong to one and only one nest. The sum of edges incident to leaf nodes must sum to unity:
\begin{equation}
    \sum_{e \in \delta_a^{in}} x_e =1 \; \; \forall a \in \mathcal{C}
\end{equation}
Next, if a nest node is included, it must have exactly one parent:
\begin{equation}
    \sum_{e \in \delta_v^{in}} x_e =y_v \; \; \forall v \in \mathcal{N}
\end{equation}
Furthermore, if a nest node is included it must have a directed edge to at least two nodes (so that it is not degenerate), and if it is not included it cannot make connections to other nodes:
\begin{equation}
    2\cdot y_u \leq \sum_{e \in \delta_u^{out}} x_e \leq |\mathcal{C}|-1\cdot y_u \; \; \forall u \in \mathcal{N}
\end{equation}
Similarly, the root node can not be degenerate and must connect to at least two nodes in the tree
\begin{equation}
    2 \leq \sum_{e \in \delta_r^{out}} x_e 
\end{equation}
If the tree has nest nodes, at least one nest node should be connected to the root:
\begin{align}
    1-(1-\delta) \leq \sum_{u \in \mathcal{N}} x_{ru}\\
    \sum_{u \in \mathcal{N}} y_u \leq |\mathcal{C}|-2 \cdot \delta,
\end{align}
where $\delta \in \{0,1\}$ is a decision variable equal to one if there is at least one nest node in the tree.\\
Choice nodes must be leaf nodes and cannot have originating edges:
\begin{equation}
    \sum_{e \in \delta_a^{out}} x_e =0 \; \; \forall a \in \mathcal{C}
\end{equation}
The root node $r$, by definition, can not have incident edges:
\begin{equation}
    \sum_{e \in \delta_r^{out}} x_e =0 \; \;
\end{equation}
Finally, we disallow self-arcs:
\begin{equation}
    x_{uu} =0 \; \; \forall u \in \mathcal{V}
\end{equation}
\subsubsection{Regularization Constraints}
Finally, we discuss enforcing constraints (3.4) and (3.5). The number of nests in tree $\mathbb{T}$ can be constrained, with great facility, to equal any given integer $M \in[0,p]$ through the following expression:
\begin{equation}
    \sum_{u \in \mathcal{N}} y_u = M
\end{equation}
Enforcing (3.5) is a little more involved. First we establish a relationship between nesting level and tree height. The height of a tree is the depth of its deepest node. The depth of a node is the number of edges on the longest downward path between the root $r$ and said node. We see that there is a one-to-one equivalence between nesting level and tree height. As an example, the tree shown in Figure 3 has 3 nesting levels (the root, nest $b_1$, and nest $b_2$). We can also count 3 edges on the path from $r$ to leaf node $a_4$ which is the deepest node of that tree, so the tree shown has height 3. To enforce a specific tree height $L$, there has to be (i) at least one path between the root $r$ and some leaf node $a\in \mathcal{C}$ with $L$ edges and (ii) no other path between the root and any node has strictly more than $L$ edges. 

Starting with (i), we notice that since any path with more than one edge from the root to a leaf node must pass through some intermediate nest nodes, the requirement that ``at least one downward path between the root $r$ and some leaf node $a\in \mathcal{C}$ has $L$ edges" can be expressed as ``at least one downward path between the root $r$ and some nest node $v\in \mathcal{N}$ has $L-1$ edges". Since the nests do not have any intrinsic labels, we can enforce a path of this length, without loss of generality, using any $L-2$ nest nodes. For convenience we choose the lexicographically first $L-2$ nest nodes. Suppose the nest node labels are ordered from $nest_1,...,nest_p$, we enforce (i) as follows:

\begin{align}
    x_{r,nest_1}=1, x_{nest_i,nest_{i+1}}=1 \; \; i=1,...,L-3
\end{align}

Next, enforcing (ii) in a direct way requires enumerating all possible paths from the root to the leaf nodes with more than $L$ edges and excluding those paths from the feasible set. Clearly there is an exponential number of such paths and the direct method is untenable. Instead we enforce (ii) in an indirect way using ``lazy constraints". In short the optimization problem is first solved without enforcing (ii). If the solution includes a violating path, such path is excluded and the problem is resolved. More details can be found in Section 4.2.2.
\subsection{Regularization}\label{ch1:opts}
Regularization is a key concept in machine learning techniques. The idea is to appropriately penalize complexity in the model to avoid fitting to noise (i.e., over-fitting). There are two questions here: \textit{what }to penalize and by \textit{what amount}. The optimal amount of penalty can be decided by evaluating the model subjected to various levels of regularization on a hold-out dataset. A model with an optimal amount of penalty does best on data the training procedure hasn't seen. This is a technique in machine learning known as cross-validation; see \cite{bishop2006pattern} for more details. In regards to what to penalize, there are two ways of making a nested logit model more complicated (i) increasing the number of nests (leading to more scale parameters to be estimated) (ii) increasing the nesting level (resulting in more non-zeros in the variance covariance matrix of the error terms).

Typically the training likelihood is a non-decreasing function of complexity. For example, adding regressors to a linear regression model cannot worsen the training likelihood. This is because the training procedure can simply set the coefficients of the added regressors to zero and obtain the same likelihood as a model without the additional regressors.
In the nested logit structure learning problem however it turns out that we should \textit{not} expect any trends on the likelihood of the training dataset. This is somewhat counter-intuitive and is due to the fact that we disallow degenerate nests (see constraint 3.16). For more details, refer to Section A.1 of the appendix.
      \section{Solution Algorithm}
 \setcounter{equation}{0}
 
Section 3 introduced a formulation of the nested logit structure learning problem (NLSLP) as a mixed integer non-linear program with linear constraints. In this section we find a practical solution algorithm to the NLSLP as introduced. The NLSLP brings with it several challenges which ultimately shape our approach to finding a practical solution method.

First, the likelihood function can only evaluated at tree solutions. This is because the inclusive values (3.10) are defined recursively, and therefore the presence of any cycles will introduce circular references. Furthermore, since the likelihood function, as written in (3.8), is defined in terms of all possible tree paths from the root to each of the leaf nodes, it is not possible to explicitly load the entire function on a computer for problems of practical size. We consider a very efficient method in Section 4.2.1 for evaluating the likelihood function at a tree solutions \textit{without} enumerating all possible paths from the root as the closed form expression for the likelihood might initially suggest. The inclusive values introduce another complication, namely the coupling together of all the model parameters. This precludes the possibility of using local search algorithms that rely on evaluating the effect on the likelihood of the inclusion or exclusion of an edge. In other words, the likelihood can not be decomposed by ``edge effects". If that were the case, it would suffice to use, for example, a maximum spanning tree algorithm to arrive at the optimal structure as in \cite{janhunen2017learning}.

Second, the likelihood function we seek to maximize (3.1) is jointly non-concave in the discrete and continuous decision variables. Once the discrete decision variables are fixed, however, the problem is reduced to the usual nested logit model estimation, which has been studied extensively in the econometrics literature (\citealp{hensher2002specification}) and (\citealp{brownstone1989efficient}), with the addition of the linear scale constraints for which several optimization techniques already exist including Lagrange multipliers methods (\citealp{bertsekas1982constrained}) and conjugate gradient methods (\citealp{goldfarb1968conjugate}). 

Third, the number of constraints is exponential in the number of nodes of the graph. Recall that the cycle elimination  constraints (3.12) are applied to every subset of the nodes of the graph. Furthermore, our regularization framework requires enforcing a specific tree height which also involves, in theory, an exponential number of constraints. We discuss an efficient workaround  using ``lazy constraints" in Section 4.2.2.

Finally, it is not possible to evaluate the likelihood function when the discrete variables are relaxed (i.e., allowed to take on continuous values). Methods in the literature of obtaining concave relaxations of functions rely on the ability to evaluate the likelihood at relaxations of the discrete variables; see for example \cite{barton-mcormick}. In the NLSLP however, the discrete decision variables represent structure and can not be relaxed-- an edge or a nest is either present or not.

In lieu of the above, we find that the most appropriate solution methodology is through a variation of the linear outer approximation algorithm introduced by \cite{duran1986outer} and later developed by \cite{fletcher-1994}. We view the discrete decision variables as \textit{complicating variables} and instead of solving the optimization problem in ``one go'', we iteratively solve two easier sub-problems. The first sub-problem deals with estimating the nested logit model parameters for a fixed tree. The second sub-problem finds the most promising tree structure to pivot to at every iteration.
We discuss this procedure at length in the next section.

\subsection{General Algorithm Overview}
Denote the cardinality of the choice set $\mathcal{C}$ by $m$. Let $q$ denote the number of parameters in the systematic component of the utility equations. Let $f(\mathbf{x},\bm{\beta, \mu})=- \mathcal{L(}\mathbf{x},\bm{\beta, \mu })$. The NLSLP can be reformulated as the following optimization problem:
    \begin{align}
    \mathbf{z^*}= \mathbf{\min_{\mathbf{x},\mathbf{y},\bm{\beta, \mu}}} \; & f(\mathbf{x},\bm{\beta, \mu})\\
    \mathbf{s.t. \;} & \mathbf{x},\mathbf{y} \in \mathcal{T} \\
    & \textbf{C}\mathbf{x}+\textbf{B}\bm{\mu} \leq \textbf{d} \\
    & \bm{\mu} \in \mathbbm{R}^{m-1}, \bm{\beta} \in \mathbbm{R}^q \\
    &\textbf{x} \in \{0,1\}^{2m-1 \times 2m-1}\\
    &\textbf{y} \in \{0,1\}^{m-2}
\end{align} 
Where $\mathcal{T}$ is the set of binary vectors $(\textbf{x} ,\textbf{y} )$ that satisfy the arborescence (3.11-3.12), structural (3.14-3.22) and regularization constraints (3.23) and (3.5), and for some suitably defined matrices \textbf{C} and \textbf{B} and vector \textbf{d} that describe the scale constraints (3.13).

Now, for any feasible tree solution $\mathbf{x}^{(k)}$, we define the nonlinear sub-problem, NLP$^{(k)}$, as 
    \begin{align}
    \mathbf{z_{NLP}}(\mathbf{x}^{(k)})= \mathbf{\min_{\bm{\beta, \mu}}} \; & f(\mathbf{x}^{(k)},\bm{\beta, \mu})\\
    \mathbf{s.t. \;} & \textbf{C}\mathbf{x}^{(k)}+\textbf{B}\bm{\mu} \leq \textbf{d }\\
    & \bm{\mu} \in \mathbbm{R}^{m-1}, \bm{\beta} \in \mathbbm{R}^q
\end{align} 
The nonlinear sub-problem NLP$^{(k)}$ finds optimal parameters $\bm{\beta, \mu}$ for a given tree $\mathbf{x}^{(k)}$. This problem is simply a nested logit estimation problem with the addition of scale constraints. As the feasible set of NLP$^{(k)}$ is a subset of the feasible set of the original problem, we have that for all $\mathbf{x}^{(k)} \in \mathcal{T}$,
 \begin{equation}
    \mathbf{z^*} \leq \mathbf{z_{NLP}}(\mathbf{x}^{(k)})
\end{equation}
In other words, the solution to any of the non-linear sub-problems NLP$^{(k)}$  provides a rigorous upper bound on the objective function value of $\mathbf{z^*}$. We refer to this problem as the upper bounding sub-problem. 

Next, we approximate the function $f$, as the maximum of its linear approximations around a set of feasible solutions $\mathcal{O}^{(k)} =\{(\mathbf{x}^{(1)},\bm{\beta}^{(1)},\bm{ \mu}^{(1)}),...,(\mathbf{x}^{(k)},\bm{\beta}^{(k)},\bm{ \mu}^{(k)})\}$.
If $f$ were a convex function then the following problem, called the ``master" mixed-integer linear program (MILP), always provides a lower bound to the original optimization problem, i.e.,
 \begin{equation}
    \mathbf{z_{MILP}}^{(k)} \leq \mathbf{z^*}
\end{equation}
The ``lower bounding" MILP master problem is given by:
\begin{align}
   \mathbf{z_{MILP}}^{(k)}= &\mathbf{\min_{\eta, \mathbf{x},\mathbf{y},\bm{\beta, \mu}}} \;  \eta\\
    \mathbf{s.t. \;}& \eta \geq f(\mathbf{x}^{(i)},\bm{\beta}^{(i)},\bm{ \mu}^{(i)}) + \nabla f(\mathbf{x}^{(i)},\bm{\beta}^{(i)},\bm{ \mu}^{(i)})^T \begin{bmatrix}\mathbf{x}-\mathbf{x}^{(i)} \\ \bm{\beta}-\bm{\beta}^{(i)}  \\ \bm{\mu}-\bm{\mu}^{(i)}\end{bmatrix} \forall (\mathbf{x}^{(i)},\bm{\beta}^{(i)},\bm{ \mu}^{(i)}) \in \mathcal{O}^{(k)}\\
            & \mathbf{x} \in \mathcal{T} \\
& \textbf{C}\mathbf{x}+\textbf{B}\bm{\mu} \leq \textbf{d }\\
    & \bm{\mu} \in \mathbbm{R}^{m-1}, \bm{\beta} \in \mathbbm{R}^q \\
       &\textbf{x} \in \{0,1\}^{2m-1 \times 2m-1}\\
    &\textbf{y} \in \{0,1\}^{m-2}
\end{align}

The representation above is the so-called epigraph formulation, where the objective function $f$ is moved out of the objective into the feasible set. If $f$ is convex the linearizations around the set of points $\mathcal{O}^{(k)}$ overestimate the feasible region and we obtain a lower bound on the objective function value as stated in (4.1). $f$ is not convex (since $-f= \mathcal{L}$ is not concave), and the tangent hyper-planes (4.13) are not necessarily global under-estimators of $f$. Consequently, the MILP problem above may cut off regions of the feasible space. An established heuristic to overcome this is to allow the linearizations to move away from the feasible region. This is done through the use of artificial non-negative variables that are penalized in the objective; see \cite{grossmann-1990}.  

In summary, the linear outer approximation solves the original optimization problem (4.1)-(4.6) by iteratively solving a sequence of two easier problems: an MILP master problem (4.12)-(4.18) and an NLP sub-problem (4.7)-(4.9).
Algorithm 1 describes the steps of the linear outer approximation procedure as it applies to the NLSLP.
\begin{algorithm}
\SetAlgoLined
 \setcounter{bean}{0}
       \begin{center}
\begin{list}
{\textsc{Step} \arabic{bean}.}{\usecounter{bean}}   
    \item Find an initial tree solution  with $M$ nests and $L$ nesting levels, call it $\mathbf{x}^{(1)}$.
    \item Estimate the taste $\bm{\beta}^{(1)}$ and scale $\bm{\mu}^{(1)}$ parameters for the tree $\mathbf{x}^{(1)}$ found in step 1 by solving NLP$^{(1)}$. 
    \item Form a set of visited solutions $\mathcal{O}^{(1)} =\{(\mathbf{x}^{(1)},\bm{\beta}^{(1)},\bm{ \mu}^{(1)})\}$. Evaluate $f(\mathbf{x}^{(1)},\bm{\beta}^{(1)},\bm{ \mu}^{(1)})$ and $\nabla f(\mathbf{x}^{(1)},\bm{\beta}^{(1)},\bm{ \mu}^{(1)})$  and form the linearization constraints (4.13). 
    \item Solve the Master MILP$^{(1)}$ (4.12)-(4.18) to obtain a new solution  $\mathbf{x}^{(2)}$ .
    \item Solve NLP$^{(2)}$, augment the set of visited solutions with the newly optimal solution $\mathcal{O}^{(2)}=\mathcal{O}^{(1)} \cup \{(\mathbf{x}^{(2)},\bm{\beta}^{(2)},\bm{ \mu}^{(2)})\}$, form the linearization constraints (4.13)
    \item Solve MILP$^{(2)}$ to pivot to a new solution.
    \item Continue iterating between NLP$^{(k)}$ and MILP$^{(k)}$ until the MILP problem is infeasible or a termination criteria is met.
    \item Evaluate the likelihood (3.8) on training and validation data-sets.
       \item Repeat steps 1 to 8 for all feasible combinations of $M$ and $L$.
       \item Choose the specification with the best validation likelihood.
 \item Estimate a nested logit model with optimal nesting structure specification as determined from step 10 on the full dataset.
    \vspace{-\baselineskip}\mbox{}
    \end{list}
       \end{center}
 \caption{Linear Outer Approximation for the NLSLP}
\end{algorithm}

A few remarks on Algorithm 1 are in order.
It is easy to find an initial tree solution with the desired number of nests and nesting levels to start the algorithm in step 1. Algorithmically, this can be done by modifying the MILP by replacing (4.13) with the constraint $\eta \geq 0$ and removing constraints (4.15) and (4.16). 

Whenever the MILP is solved additional constraints are added to cut-off previously found trees and all trees in their ``equivalence class" to guarantee finite convergence and prevent cycling behavior. In our implementation the nests are labeled, however the nest labels have no effect on the likelihood. When cutting a previously visited tree, one must also cut, from the feasible set, all trees in its equivalence class, i.e., all trees such that when the nest labels are removed, the resulting tree structure is the same. The exact form of the cuts is discussed in Section 4.2.2.

 In practice, the linear outer approximation may take a large number of iterations to converge. Typically however, the majority of the optimality gap (the difference in objective function value between the NLP and MILP) is closed during the first few iterations, and a commonly used termination criteria is iteration limit. Another criteria is the worsening of the objective function value of two successive NLP sub-problems \cite{grossmann-1986}. We use an iteration limit in our implementation. 

As linearizations are added in step 5, the master MILP becomes an improved approximation of the original optimization problem. Convergence to an optimum occurs when the value of the master MILP is worse that the value associated with the NLP subproblem, and the optimum is guaranteed to be a global optimum if the function $f$ is convex. Since $f$ is not a convex function convergence to a global optimum cannot be guaranteed.

\subsection{Practical Matters}
In this section we discuss a few practical implementation details. Crucial to the outer approximation algorithm is the ability to evaluate the value of the function $f$ and its gradients $\nabla f$ at points in the feasible set. Since the likelihood function (3.8) is defined in terms of all possible paths from the root to the leaf nodes, direct evaluation of this function is prohibitive for any choice set where the number of alternatives is not very small. Fortunately, the tree structure of the problem can be exploited as a workaround as we see in Section 4.2.1.

In Section 4.2.2, we discuss an efficient method for dealing with the exponentially many constraints in the master MILP. We end with a note on the code implementation of Algorithm 1 in section 4.2.3.

\subsubsection{Evaluating the likelihood and its gradients}
\paragraph{Efficient evaluation of the likelihood function}
Evaluating the likelihood function (3.8), ``top-down" would require enumerating all paths $l\in\{r\xrightarrow{}a\}$ for each $a \in \mathcal{C}$ - which is a prohibitive task for choice sets of practical size. In fact a short proof (see Section A.3 of the Appendix) reveals that the total number of terms in (3.8) is $|\mathcal{I}|\cdot|\mathcal{C}|\cdot\floor*{|\mathcal{C}|! \cdot e}$.
Instead, consider the following algorithm for efficiently computing the term $\Big(c_{na} \sum_{l\in\{r\xrightarrow{}a\}}x_{l} \ln{{P}(a|l)} \Big)$ in (3.8) at \textit{tree solutions} for a fixed $n \in \mathcal{I}$ and $a \in \mathcal{C}$. At such solutions, there is a single unique active path $l\in\{r\xrightarrow{}a\}$ for each $a\in \mathcal{C}$. Fix $a\in \mathcal{C}$, and denote its path by the set of nodes $b^{(1)}_{l},...,b^{(s)}_{l}$ where $b^{(1)}_{l}=r$ and $b^{(s)}_{l}=a$, where $s$ is the length of the path $l$ which we do not necessarily know \textit{a priori}. Algorithm 2 describes an efficient method of calculating the likelihood (3.8) at tree solutions:

\begin{algorithm}
\SetAlgoLined
 \setcounter{bean}{0}
       \begin{center}
\begin{list}
{\textsc{Step} \arabic{bean}.}{\usecounter{bean}}  
    \item For each $n\in \mathcal{I}$ and $a\in \mathcal{C}$ do steps 2 to 5.
    \item If $c_{na}=1$ continue to step 2, otherwise the contribution to the likelihood is zero.
    \item Start at a leaf node $a$ and propagate to the node's parent $B(a)$. Add to the likelihood, the quantity $\mu_{B(a)}V_{an}$.
    \item If the current node is the root node add the quantity $-\mu_r\Gamma_n\{r\}$ to the likelihood and stop. Otherwise, propagate to the current node's parent $B(B(a))$ and add the following quantity to the likelihood $(\mu_{B(B(a))}-\mu_{B(a)})\Gamma_n\{B(a)\}$. Set the current node to $B(a)$.
    \item Continue adding contributions as in step 4 until the root node is reached.
    \vspace{-\baselineskip}\mbox{}
    \end{list}
       \end{center}
 \caption{Efficient evaluation of the likelihood at tree solutions}
\end{algorithm}

\paragraph{Computing gradients}
Central to the linear outer approximation algorithm is the availability of gradients of the likelihood function at specified tree solutions. We make a distinction here between the continuous variables $\bm{\beta}$ and $\bm{\mu}$, and the discrete variables $\mathbf{x}$. 

Derivatives of the likelihood function $\mathcal{L}$ with respect to continuous variables can be computed analytically or through automatic differentiation (\citealp{baydin2017automatic}). Automatic differentiation however, can not reliably handle derivatives with respect to discrete decision variables. We resort to analytical differentiation and find that closed form derivatives exist, and can be efficiently evaluated at tree solutions using variations of Algorithm 2. These derivatives are somewhat involved, the complete derivations can be found in Section A.2 of the appendix.

\subsubsection{Lazy constraints and formulations}
Lazy constraints make it possible to leave out constraints from an optimization problem that must be satisfied by any valid solution but which may not be required by the solution algorithm to arrive at the optimal solution -- resulting in a smaller problem that is easier to solve. Only those constraints that are required for finding the optimal solution are included, and only when they are needed (\citealp{pearce2019towards}). 

In the NLSLP, the cycle elimination constraints and the anti-cycling constraints are exponentially sized in the cardinality of the choice set. The tree height constraints are difficult to explicitly describe but a violating solution can be easily detected and removed from the feasible set. We describe lazy constraint formulations to address these points.

\paragraph{Cycle elimination constraints}
The cycle elimination constraints (3.12) are necessary to avoid passing non-tree solutions to the NLP sub-problem. There is one such constraint for every strict subset of the nodes of the graph, consequently if a graph has $N$ nodes there is $2^N-1$ of these constraints. Generating and adding these constraints at once to the MILP is not practical. Instead we take a lazy constraint approach. Since it is easy to detect the presence of cycles in current solutions and produce constraints to remove such solutions from the feasible set we generate and add only the required cycle elimination constraints ``on the fly" as needed. 

The MILP problem is first solved without the cycle elimination constraints. A ``separation oracle" is then used to determine which of cycle elimination constraints are violated. \textit{Only} the violated constraints are added and the MILP is re-solved. It is important that the separation oracle is efficient at detecting constraint violations. A naive exponential time oracle will generate and check each of the cycle elimination constraints for violations. Instead, it suffices to run a depth first traversal of the graph and check if that yields any back edges. This algorithm has a polynomial worst-case run time; see \cite{cormen2009introduction}. 

\paragraph{Anti-cycling constraints} Since nest labels have no effect on the likelihood, there is in fact an equivalence class of tree solutions. As an example Figure 4 shows two trees representing the same nesting structure with different labeling of the nests. Visiting either of the trees necessitates removing the other from pivot consideration in step 6 of Algorithm 1. In general suppose a tree solution has $N$ nests, then there exists $N!$ trees in its equivalency class (corresponding to all possible permutations of the nest labels). It is therefore crucial to avoid pivoting to trees that belong to the equivalence classes of any previously visited solution (otherwise Algorithm 1 may cycle through all $N!$ trees for a given class before pivoting to a truly new solution). All trees in the same equivalence class share a common ``signature" -namely the ancestry of the elementary alternatives. The ancestry of a current solution is used to determine if it belongs to the equivalence class of a tree that has been previously visited. That solution is then removed from the feasible set (4.14) accordingly through the constraints we now describe.

Suppose we wish to exclude a particular tree solution $\textbf{x}$ from the feasible set. Define the index sets 
$O =\{i: x_i=1\}$ and $Z=\{i: x_i=0\}$.
The following constraint removes the solution $\textbf{x}$ from the feasible set:
\begin{align}
    \sum_{i\in O}x_i - \sum_{i \in Z} x_i\leq |O| -1
\end{align}
Again, we take a lazy constraint approach. These cuts are not generated at once but are instead added on the fly as needed. 
\paragraph{Tree height constraints} The regularization framework (3.4)-(3.5) of the NLSLP involves limiting the nesting levels (or equivalently the height of the nesting tree). We noted in section 3.2.4 that to enforce a specific tree height $L$, we require that (i) at least one path between the root $r$ and some leaf node $a\in \mathcal{C}$ has $L$ edges and (ii) no other path between the root and any node has strictly more than $L$ edges. We elaborate on the lazy constraint approach we apply to enforcing (ii).  
At a given tree solution we traverse to the root node starting from each of the elemental alternatives counting the number of edges along the way. If a path of length strictly greater than $L$ is found, the tree solution is eliminated from the feasible set using constraints (4.19).

\subsubsection{Implementation}
The nested logit structure learning problem has been implemented in the Julia programming language (\citealp{julia-2017}). The Gurboi solver (\citealp{gurobi2015gurobi}) is used to solve the MILP master problems. The JuMP interface (\citealp{DunningHuchetteLubin2017}) in Julia allows for user-defined lazy cuts, which is critical to our implementation as described in section 4.2.2. IPOPT (\citealp{wachter2006implementation}) is used to solve the NLP sub-problems. IPOPT applies an interior point algorithm to solve the linearly constrained NLP sub-problems. The authors make the source code accessible under an MIT licence through \url{github.com/ymedhat95/nested-logit.git}.

       \section{Computational Experiments}
 \setcounter{equation}{0}
 In this section we validate Algorithm 1 introduced in Section 4.1. The algorithm is used to learn an optimal nesting tree specification empirically from the data. In Section 5.1, we demonstrate, through a Monte Carlo experiment that Algorithm 1 can correctly recover the true tree structure as the sample size increases. In Section 5.2, we apply our algorithm a travel mode choice dataset from \cite{viegas2018modeling}.
 \subsection{Monte Carlo Experiments}
The purpose of this section is to demonstrate the statistical consistency of our methodology of learning a nested logit structure from the data. To this end, we simulate a synthetic choice scenario as follows. We consider a population of rational agents $n\in \mathcal{I}$ each making a decision from a set of eight alternatives $\mathcal{C}=\{a_1,...,a_8\}$. The utility of alternative $a\in \mathcal{C}$ to agent $n\in \mathcal{I}$ is given through the following utility equations:
\begin{equation}
    U_{na}=ASC_{a}+\epsilon_{na} 
\end{equation}
Where $ASC_{a}$ is an alternative specific constant for alternative $a$.
The error terms, $\epsilon_{na}$, are independent across agents but are correlated for the same individual across the different alternatives. The correlation structure of the error terms is described by the tree shown in Figure 5. The nesting tree and associated scale parameters shown in Figure 5 imply a variance-covariance matrix for the joint distribution of the error terms, $\epsilon_{na_1},...,\epsilon_{na_8}$ for a given agent. This variance-covariance matrix and values for the alternative specific constants are shown in Table 1. For identification purposes the ASC of the first alternative is normalized to zero. Two blocks of non-zero covariances (highlighted) can be seen in Table 1. The goal is to use Algorithm 1 to correctly recover the values for the ASCs and the elements of the variance-covariance matrix. 

Synthetic data was generated according to the choice scenario and parameters just described for 35 repetitions for various sample sizes of agents: (a) 25,000 (b) 10,000 and (c) 5000. For every repetition we used Algorithm 1 to learn the nesting tree, its parameters and estimate the ASCs from the data. We then used Proposition 2.1 to obtain the equivalent variance-covariance matrix representation of the estimated tree structure and scale parameters. The estimated variance-covariance matrix and ASCs where then averaged over the 35 different repetitions to obtain the means and standard errors for the various sample sizes. 

The results are shown in Table 2. For the largest sample size, (a), we can see that Algorithm 1 correctly recovers the underlying tree structure and its parameters in all 35 repetitions. For (b) the algorithm sometimes misses nests $n_1$ and $n_2$ thus shrinking some of the significant covariances. Once or twice, the algorithm captures some spurious correlations as indicated by the very small non-zeros shown in the unshaded blocks with large standard errors. Similar but more pronounced effects are observed for the smallest sample size (c).

\subsection{Empirical Application}
Data from the 2010 Massachusetts Travel Survey (MTS) and
matrices for car and transit travel times and costs, provided by Boston’s Central Transportation Planning Staff (CTPS) (\citealp{de2018modelling}), are used to estimate a nested logit choice model for the work travel mode. Individuals were asked to fill out all activities performed on a designated weekday, and to provide the activity location and the transport mode used to arrive at this location. The survey also collected individual and household characteristics for participants. There are six main travel modes reported in the survey: \textit{Walk}, \textit{Bike},\textit{ Car}, \textit{Car Pool (2 people)}, \textit{Car Pool (3+ people)}, and \textit{Transit}.

The systematic utilities for each of these travel modes are specified in Table 3. The specification is linear in parameters and includes a constant (intercept). The travel mode attributes travel time, travel cost, an indicator of whether the trip is made to the central business district (CBD) of Boston as well as traveler's income and gender also enter the utility equations. There is an alternative specific coefficient $\beta$ for each travel mode attribute and each traveler characteristic. For purposes of identification, all alternative specific coefficients for the walk alternative (except the travel time coefficient) are fixed to zero. All six alternative specific travel time coefficients are identifiable because the travel time varies over the six alternatives, this is not true of the other variables (e.g. income). A normalization is therefore needed. 

The goal is to estimate a nested logit model (structure and parameters), consistent with utility maximization, from the data. Algorithm 1 is used. Table 4 shows the training and validation negative log-likelihood at converged solutions for all feasible combinations of the regularization parameters $M$ and $L$ in Algorithm 1. The best performing model on the validation dataset is the tree with $4$ nests and $5$ nesting levels shown in Figure 6 (left). The model with the best validation log likelihood is chosen. The nested logit model was then estimated for this tree on the full dataset. The estimated parameters are shown in Table 5.

The estimated alternative specific constants $C$ show that, \textit{ceteris paribus}, the bike, car pooling, and transit travel modes are less preferred than the walk mode, while the car mode is slightly more preferred. Travelers are most sensitive to walk and bike travel times and least sensitive to car travel time. The possession of a transit pass makes a traveler more likely to bike or take transit and less likely to travel by car or car pool than walk. Trips to the CBD are more likely to be made by transit and less likely to be made by car, car pooling or bike relative to walk. 

The complicated nesting structure in Figure 6 (left) reveals nonzero correlations between the error terms in the utility equations of all the travel modes except for the $carpool2$ mode. Looking at the estimated scale parameters on the full dataset shown in Table 5 we notice that the scale parameter for $n_1$, 1.215, is not significantly different from the scale parameter for the root node which is normalized to 1. Consequently $n_1$ can be collapsed to the root node. Similarly the scale parameter for $n_3$ is not statistically different from the scale parameter of $n_2$ and the former nest can be collapsed into the latter. The resulting tree shown in Figure 6 (right) is simpler and an attempt at ``labeling" the nests has be made. The final nesting tree groups the car, transit, walk, and bike modes in a nest ``Traditional" modes of transportation. Within that nest, the walk and bike modes are grouped into another nest for ``Active" travel modes.

Figure 7 shows two alternative models that would typically be estimated for a travel model choice on the basis of intuition. The first model, Model A, groups the walk and bike travel modes into an ``Active" nest and the car and car pooling modes into an ``Auto" nest. The second model, Model B, groups transit and the car pooling modes into a ``Shared" nest and the walk, bike, and car travel modes in to an ``Unshared" nest. Table 6 compares the fit on the data of these two models to the one obtained through Algorithm 1. The multinomial logit (no nesting) model is used as a benchmark. 

It can be seen that the model obtained through Algorithm 1 provides a much superior fit than the rest of the models. A total of 45 nesting trees visited during the search for the this tree (this is the number of NLP sub-problems solved in Algorithm 1), a small fraction of the 2712 possible non-degenerate nesting trees that could be formed with 6 alternatives. 

       \section{Concluding Remarks}
 \setcounter{equation}{0}
 
 Since its introduction, the nested logit model has been found to be extremely flexible with a myriad of applications in economics, marketing, and transportation. A large number of alternative nesting structures are possible in any choice context in which the number of alternatives is not very small. An appropriate method for systematically choosing an optimal nesting structure has not yet appeared in the literature. The dependency of the optimal structure on the systematic specification complicates this task yet further. While \textit{a priori} reasoning may be of some help, results in the literature have suggested that intuition is an imperfect guide. 
 
 In this work we provide an optimization-based framework for learning an optimal nesting structure from the data. We formally introduced the nested logit structure learning problem (NLSLP) and provided a complete solution algorithm in addition to an open-source code implementation available to practitioners. Our holistic view of nested logit estimation takes into consideration the dependency of the optimal structure on the systematic specification of the utility equations. While the optimization may appear intractable at first glance, we have found ways to exploit the tree structure of the problem and utilized the state-of-the art in algorithmic development in solving mixed integer optimization problems and lazy constraints to bring practical tractability to the NLSLP.
 
We have demonstrated the efficacy of our algorithm by applying it to a synthetic dataset where the correct nesting tree structure was successfully recovered. Finally, we demonstrated an empirical application to a travel mode choice dataset.  


    
    \bibliography{references} 
\newpage

\begin{table}[h]
\begin{center}
\caption{Parameter values of the alternative specific constants and the variance-covariances of the error terms used to generate the synthetic data.}
\centering
\includegraphics[scale=0.45]{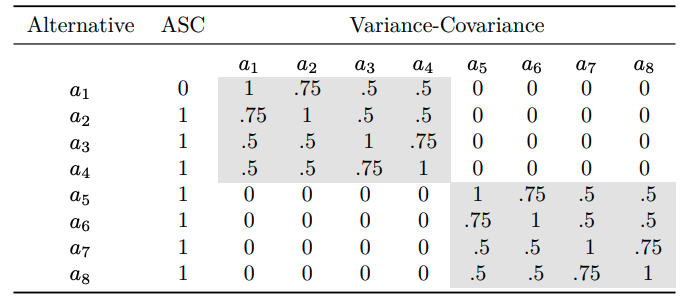}
\end{center}
\footnotesize
\renewcommand{\baselineskip}{11pt}
\textbf{Note:} {The variance-covariance values are given by $\frac{\pi^2}{6}$ times the entries shown in the table (c.f. Proposition 2.1).}
\noindent
\end{table}

\begin{table}
\begin{center}
  \caption{ Mean parameter values and standard errors for the alternative specific constants and the recovered variance-covariance matrix for different sample sizes.}

\caption*{(a) $N_{train}=20000$, $N_{validation}=5000$}
\includegraphics[scale=0.45]{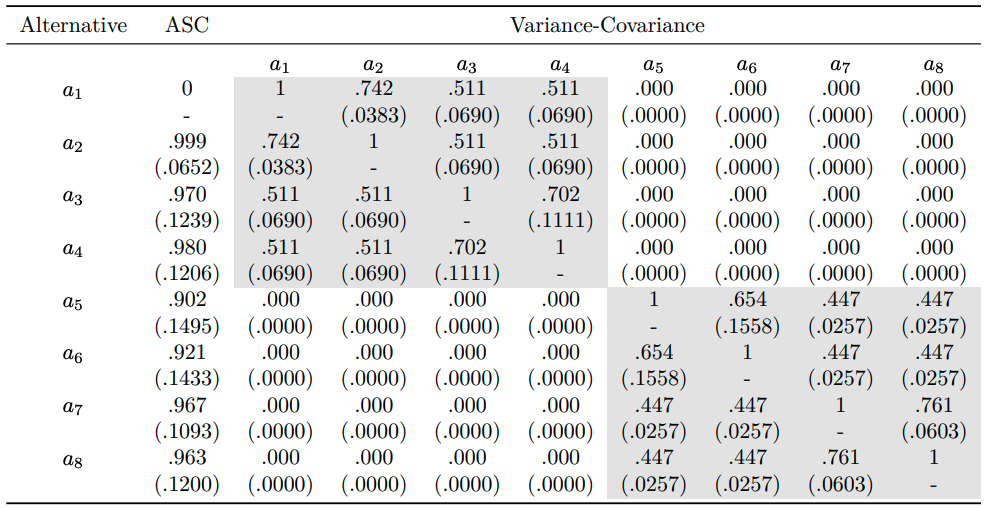}
\end{center}
\end{table}

\begin{table}[h]
\begin{center}
  \caption*{{\textbf{Table 2 }(continued):} Mean parameter values and standard errors for the alternative specific constants and the recovered variance-covariance matrix for different sample sizes.}

\caption*{(b) $N_{train}=7500$, $N_{validation}=2500$}
\includegraphics[scale=0.45]{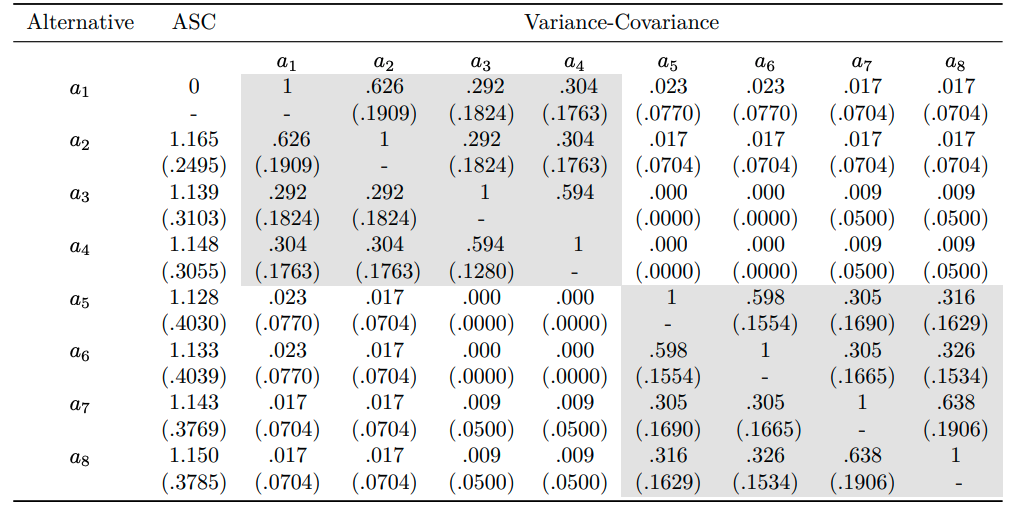}
\medskip

\caption*{(c) $N_{train}=3750$, $N_{validation}=1250$}
\includegraphics[scale=0.45]{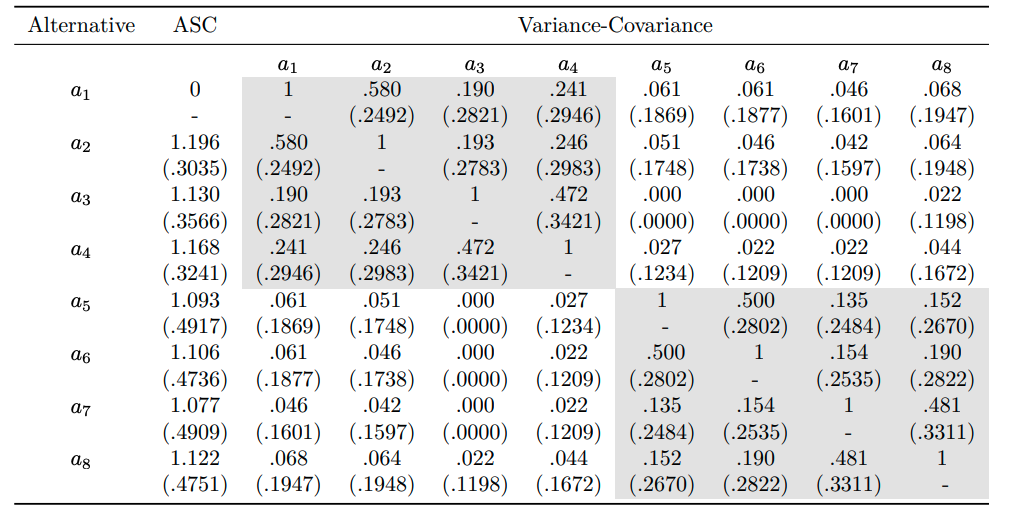}
\end{center}

\footnotesize
\renewcommand{\baselineskip}{11pt}
\textbf{Note:} { Means and standard errors where calculated across 35 repetitions for each sample size. Standard errors are shown in brackets below the corresponding mean. The variance-covariance values are given by $\frac{\pi^2}{6}$ times the entries shown in the table (c.f. Proposition 2.1). }
\noindent
\end{table}

\newpage

\begin{table}[h]

\caption{Systematic specification of the utility equations for the work travel mode choice model.}
\begin{center}

\begin{tabular}{cccccccccc}
\toprule
\multicolumn{1}{c}{} & \multicolumn{1}{c}{} & \multicolumn{7}{c}{Attributes \& Characteristics} \\

\multicolumn{1}{c}{} &  & Constant & Travel Time & Travel Cost & Transit Pass & CBD Trip & Income & Female  \\
\multicolumn{1}{c}{} & & $C$ & $\beta_{tt}$ & $\beta_{tc}$ & $\beta_{tp}$ & $\beta_{cbd}$ & $\beta_{inc}$ & $\beta_{female}$ \\
\midrule
\multirow{6}{*}{\rotatebox[origin=c]{90}{Alternative}}  & $walk$ & 0 &\checkmark & - & 0 & 0 & 0 & 0  \\
 & $bike$ & \checkmark & \checkmark & - & \checkmark & \checkmark & \checkmark & \checkmark   \\
 & $car$ & \checkmark & \checkmark & \checkmark & \checkmark &\checkmark &\checkmark & \checkmark \\
 & $carpool2$ & \checkmark& \checkmark & \checkmark & \checkmark & \checkmark & \checkmark &\checkmark  \\
 & $carpool3+$ &\checkmark&\checkmark & \checkmark& \checkmark & \checkmark &\checkmark & \checkmark \\
 & $transit$ &\checkmark & \checkmark &\checkmark& \checkmark & \checkmark & \checkmark &\checkmark \\
 \bottomrule
\end{tabular}
\end{center}

\footnotesize
\renewcommand{\baselineskip}{11pt}
\textbf{Note:} {A check mark indicates that the attribute or characteristic enter the utility function of the corresponding alternative linearly. A zero indicates that the corresponding parameter is fixed to zero for identification. All parameters are alternative specific.}
\end{table}

\begin{table}[h]
\caption{{Training and validation log-likelihoods at converged solutions for all feasible combinations of number of nests and nesting levels. The best validation log-likelihood is underlined.}}
\parbox{.45\linewidth}{
\centering
\caption*{Training Negative Log-likelihood $-\mathcal{L}$}
\includegraphics[scale=0.45]{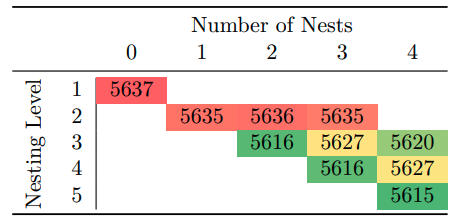}
}
\hfill
\parbox{.45\linewidth}{
\centering

\caption*{Validation Negative Log-likelihood $-\mathcal{L}$}
\includegraphics[scale=0.45]{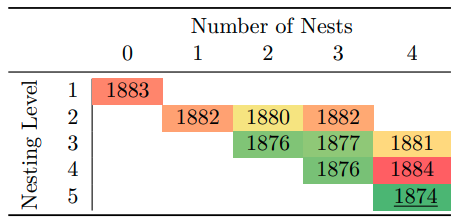}
}

\end{table}

\begin{table}[h]

\caption{Estimated parameters for the optimal nesting structure for work mode choice model.}
\begin{center}

\begin{tabular}{cccccccccccc}
\toprule
\multicolumn{1}{c}{} & \multicolumn{1}{c}{} & \multicolumn{10}{c}{Estimated Parameters} \\

\multicolumn{1}{c}{} &  & $C$ & $\beta_{tt}$ & $\beta_{tc}$ & $\beta_{tp}$ & $\beta_{cbd}$ & $\beta_{inc}$ & $\beta_{female}$ & \multicolumn{1}{c}{} &  & $\mu$ \\
\midrule
\multirow{12}{*}{\rotatebox[origin=c]{90}{Alternative}} & $walk$ & 0 & -2.364 &  & 0 & 0 & 0 & 0 &\multirow{12}{*}{\rotatebox[origin=c]{90}{Nest}} & $n_1$ & 1.215 \\
 &  & - & (.4466) &  & - & - & - & - &  &  & (.1498) \\
 & $bike$ & -.802 & -2.587 &  & .041 & -.307 & .006 & -.284 &  & $n_2$ & 1.747 \\
 &  & (.2478) & (.5203) &  & (.0701) & (.1301) & (.0054) & (.1215) &  &  & (.1954) \\
 & $car$ & .557 & -.399 & -.101 & -1.001 & -.977 & -.017 & .128 &  & $n_3$ & 1.847 \\
 &  & (.1258) & (.3016) & (.0329) & (.1533) & (.1664) & (.0065) & (.0871) &  &  & (.2966) \\
 & $carpool2$ & -1.478 & -.443 & -.265 & -1.074 & -1.022 & .010 & .423 &  & $n_4$ & 4.350 \\
 &  & (.1467) & (.4416) & (.0924) & (.1445) & (.1742) & (.0078) & (.1055) &  &  & (1.6591) \\
 & $carpool3+$ & -1.810 & -.454 & -.420 & -1.156 & -.926 & .010 & .416 &  &  & \multicolumn{1}{c}{} \\
 &  & (.2805) & (.5165) & (.1703) & (.1886) & (.2106) & (.0085) & (.1198) &  &  & \multicolumn{1}{c}{} \\
 & $transit$ & -.783 & -.840 & -.022 & .829 & .148 & -.022 & .113 &  &  & \multicolumn{1}{c}{} \\
 &  & (.1760) & (.1415) & (.0117) & (.1394) & (.1258) & (.0072) & (.0928) &  &  & \multicolumn{1}{c}{}\\
 \bottomrule
\end{tabular}
\end{center}

\footnotesize
\renewcommand{\baselineskip}{11pt}
\textbf{Note:} {The nesting specification for the estimated model is shown in Figure 6 (left) and the systematic specification is shown in Table 5. The scale parameter of the root node, $\mu_r$, is normalized to 1. Standard errors are shown in brackets below the corresponding estimate.}
\noindent
  \end{table}

\begin{table}[h]
\caption{Comparison of fit on the work travel mode choice dataset between different nesting specifications.}
\begin{tabular}{cccccc}
\toprule
\multirow{2}{*}{{Tree}} & \multirow{2}{*}{{No. of Nests}} & \multicolumn{3}{c}{{Negative Log-likelihood $-\mathcal{L}$}} & \multirow{2}{*}{{No. of Visited Trees}} \\
 &  & {Full dataset} & {Training} & {Validation} &  \\
 \midrule
Algorithm 1& 4 & 7492 & 5615 & 1874 & 45 \\
\textit{Figure 6 (left)} & & & & &\\
Model A & 2 & 7523& 5642 & 1887 & 1 \\
\textit{Figure 7 (left)} & & & & &\\
Model B & 2 & 7487 & 5617 & 1888 & 1 \\
\textit{Figure 7 (right)} & & & & &\\
Multinomial& 0 & 7511 & 5637 & 1883 & 1 \\
\midrule
\multicolumn{2}{c}{Sample Size} & 8936 & 6702 & 2234 & Total Possible: 2712\\
\bottomrule
\end{tabular}

\footnotesize
\renewcommand{\baselineskip}{11pt}
\textbf{Note:} {The training to validation split ratio is 3:1. The validation log-likelihood is obtained by evaluating the log-likelihood (3.8) on the validation data using parameters estimated on the training data.}
\noindent
\end{table}

\begin{figure}[h]
    \scalebox{0.9}{
  \begin{minipage}[b]{0.25\linewidth}
    \centering
\includegraphics[scale=0.48]{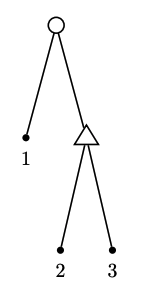}  \caption*{ $\{\{2,3\},\{1,2,3\}\}$\\}
  \end{minipage}

  \begin{minipage}[b]{0.25\linewidth}
    \centering
\includegraphics[scale=0.5]{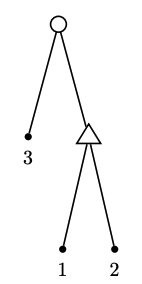} 
    \caption*{ $\{ \{1,2\},\{1,2,3\}\}$\\}
  \end{minipage}
  
    \begin{minipage}[b]{0.25\linewidth}
    \centering
\includegraphics[scale=0.5]{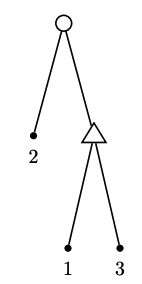} 
   \caption*{ $\{\{1,3\},\{1,2,3\}\}$\\}
  \end{minipage}
    \begin{minipage}[b]{0.25\linewidth}
    \centering

\includegraphics[scale=0.5]{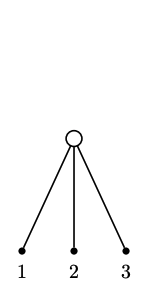} 
    \caption*{ $\{\{1,2,3\}\}$}
  \end{minipage}
  }
      \caption{The four possible non-degenerate nesting partitions for the set $\{1,2,3 \}$.}
\end{figure}

\begin{figure}[h]

\begin{minipage}{.5\textwidth}
    \centering

\includegraphics[scale=0.55]{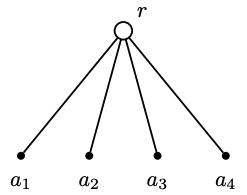} 
  
\end{minipage}
\begin{minipage}{.5\textwidth}

 \[\Sigma = \frac{\pi^2}{6\mu_r^2}
\begin{blockarray}{ccccc}
a_1 & a_2 & a_3 & a_4 & \\
\begin{block}{(cccc)c}
  1 & 0 & 0 & 0 &  a_1 \\
  0 & 1 & 0 & 0 &  a_2 \\
  0 & 0 & 1 & 0 &  a_3 \\
  0 & 0 & 0 & 1 &  a_4 \\
\end{block}
\end{blockarray}
 \]
\end{minipage}
\caption{The nested partition $\mathcal{B}=\{\{a_1,...,a_4\}\}$ over the set $\mathcal{C}=\{a_1,...,a_4\}$ and the equivalent nesting tree and variance-covariance matrix of the joint distribution of the error terms $\epsilon_{a_1},...,\epsilon_{a_4}$ for a multinomial logit model.}

\end{figure}

\begin{figure}[h]
 
\begin{minipage}{.5\textwidth}
    \centering

\includegraphics[scale=0.55]{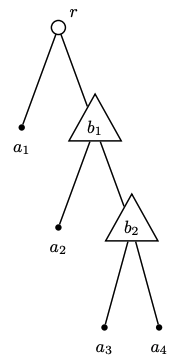} 
    
\end{minipage}
\begin{minipage}{.5\textwidth}

 \[\Sigma = \frac{\pi^2}{6\mu_{r}^2}
\begin{blockarray}{ccccc}
a_1 & a_2 & a_3 & a_4 & \\
\begin{block}{(cccc)c}
  1 & 0 & 0 & 0 &  a_1 \\
  0 & 1 & 1-\frac{\mu_{r}^2}{\mu_{b_1}^2} & 1-\frac{\mu_{r}^2}{\mu_{b_1}^2} &  a_2 \\
  0 & 1-\frac{\mu_{r}^2}{\mu_{b_1}^2} & 1 & 1-\frac{\mu_{r}^2}{\mu_{b_2}^2} &  a_3 \\
  0 & 1-\frac{\mu_{r}^2}{\mu_{b_1}^2} & 1-\frac{\mu_{r}^2}{\mu_{b_2}^2} & 1 &  a_4 \\
\end{block}
\end{blockarray}
 \]
\end{minipage}
\caption{The nested partition $\mathcal{B}=\{\{a_1,a_2,a_3,a_4\},\{a_1\},\{a_2,a_3,a_4\},\{a_3,a_4\}\}$ of the set $\mathcal{C}=\{a_1,a_2,a_3,a_4\}$ and the equivalent nesting tree and variance-covariance matrix of the joint distribution of the error terms $\epsilon_{a_1},...,\epsilon_{a_4}$ for a nested logit model.}
\end{figure}

\begin{figure}[h]
    \centering
\includegraphics[scale=0.55]{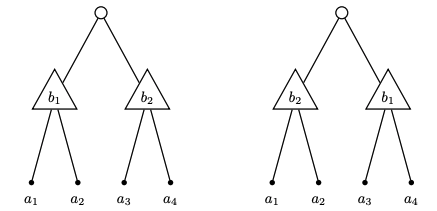}

\caption{Trees representing the same nesting structure with different nest labels. These two trees belong to the same equivalency class.}
  \end{figure}

\begin{figure}[h]
    \centering

\includegraphics[scale=0.48]{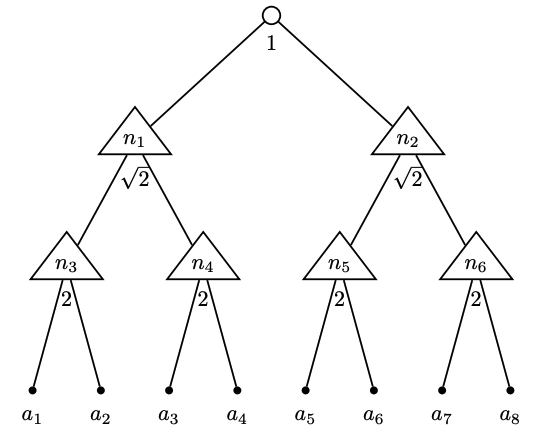} 
    \caption{The nesting tree encoding the correlation structure of the synthetic data. The scale parameters are shown below the respective nests. The root scale parameter is normalized to one.}
    \label{fig:my_label}
\end{figure}

\begin{figure}[h]
\centering

 \includegraphics[scale=0.6]{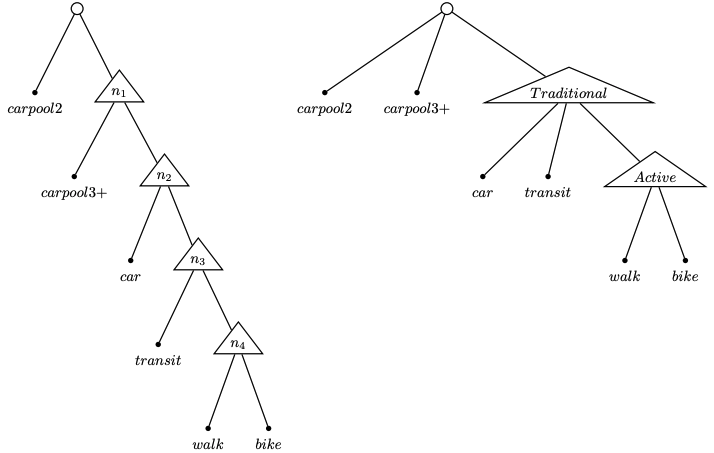}

  \caption{The learned nesting structure for the work travel model dataset (left) and a possible interpretation of the tree after collapsing nests with statistically insignificant scale parameters on the full dataset.}
\end{figure}

\begin{figure}[h]
\centering

\scalebox{0.85}{
 \includegraphics[scale=0.55]{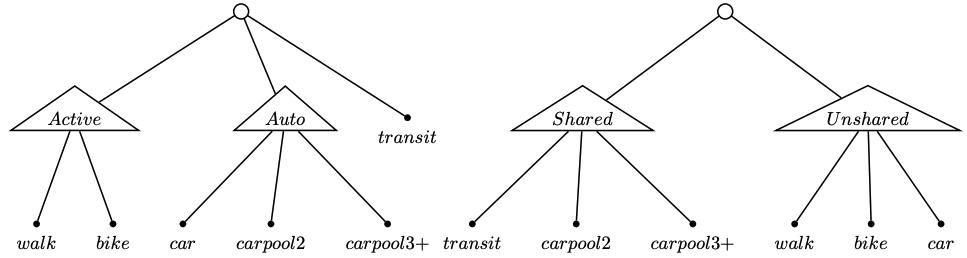} 

}

    \caption{Alternative models for the work travel mode choice dataset. Model A (left) and Model B (right)}
\end{figure}

\clearpage
   \section*{A. Appendix}
 \subsection*{A.1. More on regularization}
 We look at three possible cases and provide counter examples that show that the training likelihood can worsen by increasing complexity as defined by the number of nests and the nesting level: 

\begin{list}
{\textsc{Case} \arabic{bean}.}{\usecounter{bean}}   
  \item \textit{Increasing the number of nests by 1, while holding the nesting level constant}.
    Consider the optimal tree with $2$ nests and height $2$ shown in the left of Figure A.1. In this example, we are able to increase the number of nests by 1 without increasing the nesting level by grouping leaf nodes $3$ and $4$ in a new nest $c$. The likelihood can\textit{not} worsen, as the estimated scale of the new nest $c$ can be made equal to the normalized root scale-- effectively setting the covariance between the error terms of alternatives 3 and 4 to zero (c.f. Proposition 2.1).
  
However, it is not always possible to increase the number of nests without worsening the training likelihood. As an example, consider the optimal tree with $2$ nests and $2$ nesting levels shown in Figure A.2, and consider the addition a third nest $c$ to this tree. It is not possible to add this nest without either changing the alternative to nest allocations, increasing the nesting level, or running into degeneracy. Therefore, there can be no guarantee that the likelihood will not worsen.

    \item \textit{Increasing the nesting level by 1, while holding the number of nests constant}:
    Since increasing the nesting level without adding any additional nests would entail changing the nesting structure of a present nest or nests, there is clearly no guarantee that the likelihood cannot worsen.
    \item \textit{Increasing the number of nests and the nesting level each by 1}:
    Consider the optimal tree with a single nest and 2 nesting levels shown in the left of Figure A.3. Leaf nodes 3, 4, and 5 can be nested together in a new nest $b$ with the same scale parameter as nest $a$ increasing the number of nests and the nesting level without worsening the likelihood. \\
    However, there is no guarantee in general that the likelihood cannot worsen when both the number of nests and the nesting level are increased by one. As an example, consider the tree with 2 nests and 3 nesting levels shown in Figure A.4. The only way of increasing the nesting level of this tree is by nesting leaf nodes 4 and 5 together in one new nest \textit{c}, making nest \textit{b} degenerate.
\end{list}

  \subsection*{A.2. Derivatives of the likelihood function with respect to the edge indicators}

The discrete edge indicator variables $\textbf{x}$ can be broken down into four distinct types. The derivative of the likelihood function with respect to each of these four types has a different closed form which we derive in what follows.

(1) First order partial derivative of the likelihood function with respect to edges between the root and alternatives $x_{ra}$:
    \begin{align}
        {\frac{\partial \mathcal{L}}{\partial x_{ra}}} &= \sum_{n\in\mathcal{I}} \Big( c_{na}\big(  \ln{{P}_n\{a|\{r,a\}\}} +\sum_{l \in \{r\xrightarrow{}a\}}x_{l} \frac{\partial\ln{{P}_n\{a|l}\}}{\partial x_{ra}}\big) +\sum_{a'\in \mathcal{C}\setminus \{a\}}c_{na'}\big(\sum_{h \in \{r\xrightarrow{}a'\}}x_{h} \frac{\partial\ln{{P}_n\{a'|h}\}}{\partial x_{ra}}\big) \Big) \notag \\
        &=\sum_{n\in\mathcal{I}} \Big( c_{na}\big(\mu_rV_{an} -\mu_r\Gamma_n\{r\}\big) +\sum_{a'\in \mathcal{C}}c_{na'}\big(-\mu_r\frac{\partial\Gamma_n\{r\}}{\partial x_{ra}}\big) \Big) \notag \\
        &=\sum_{n\in\mathcal{I}} \Big( c_{na}\big(\mu_rV_{an} -\mu_r\Gamma_n\{r\}\big) -\sum_{a'\in \mathcal{C}}c_{na'}\exp{\big(\mu_r\big(V_{an}-\Gamma_n\{r\}\big)\big)} \Big) \tag{A.1}
    \end{align}
Expression (A.1) is defined in terms of summations over the elementary alternatives $\mathcal{C}$ and the set of individuals $\mathcal{I}$ and can be easily evaluated. \\
(2) First order partial derivative of the likelihood function with respect to edges between the root and nests $x_{rb}$:\\

By ``conditioning" on the nest node $b$ we can rewrite the log-likelihood as
    \begin{align*}
  \mathcal{L}  &=\sum_{n\in\mathcal{I}} \sum_{a\in\mathcal{C}} \Big(c_{na}( x_{ra}\ln{{P}_n\{a|\{r,a\}\}}+ \sum_{b\in \mathcal{N}}\sum_{l\in \{b\xrightarrow{}a\}}x_{rb}x_{l} \ln{{P}_n\{a|\{r,l\}\}}) \Big)
    \end{align*}
    Then,
    \begin{align}
       \frac{\partial \mathcal{L}}{\partial x_{rb}}&=  \sum_{n\in\mathcal{I}} \sum_{a\in\mathcal{C}} \Big(c_{na}( \sum_{l\in\{b\xrightarrow{}a\}} (x_{l} \ln{{P}_n\{a|\{r,l\}\}} - x_{rb}x_{l} \mu_r\frac{\partial\Gamma_n\{r\}}{\partial x_{rb}}) - \sum_{b'\in \mathcal{N} \setminus \{ b \}}\sum_{h \in \{b'\xrightarrow{}a\}}x_{rb'}x_{h} \mu_r\frac{\partial\Gamma_n\{r\}}{\partial x_{rb}}) \Big) \notag \\
         &=\sum_{n\in\mathcal{I}} \sum_{a\in\mathcal{C}} \Big(c_{na}(\sum_{l \in \{b\xrightarrow{}a \}}  x_{l} \ln{{P}_n\{a|\{r,l\}\}} - \sum_{h \in \{r\xrightarrow{}a \}}x_{h} \mu_r\frac{\partial\Gamma_n\{r\}}{\partial x_{rb}}) \Big) \notag \\
         &=\sum_{n\in\mathcal{I}} \sum_{a\in\mathcal{C}} \Big(c_{na}( \sum_{ l \in \{b\xrightarrow{}a \}} x_{l} \ln{{P}_n\{a|\{r,l\}\}}) - \sum_{h \in \{r\xrightarrow{}a\}}x_{h} \exp{\big(\mu_r\big(\Gamma_n\{b\}-\Gamma_n\{r\}\big)\big)} \Big) \tag{A.2}
    \end{align}
    Now, at tree solutions, there is a unique path $h \in \{r\xrightarrow{}a\}$ such that $x_{h}=1$. We can therefore do away with the fourth summation in (A.2). Finally we obtain:
    \begin{align}
         {\frac{\partial \mathcal{L}}{\partial x_{rb}}}        &=\sum_{n\in\mathcal{I}} \sum_{a\in\mathcal{C}} \Big(c_{na}( \sum_{l \in \{b\xrightarrow{}a\}}x_{b\xrightarrow{}a} \ln{{P}_n\{a|\{r,l\}\}}) - \exp{\big(\mu_r\big(\Gamma_n\{b\}-\Gamma_n\{r\}\big)\big)}  \Big)  \tag{A.3}
    \end{align}
To evaluate the contributions of the third summation efficiently at tree solutions, we need not enumerate all paths $b\xrightarrow{}a$. Instead, a simple variant of Algorithm 2 can be adapted by treating node $b$ as the root node.\\
(3) First order partial derivative of the likelihood function with respect to edges between nests and other nests $x_{bb'}$:\\
By conditioning on nests $b$ and $b'$, we can rewrite the likelihood in a more convenient form
    \begin{align*}
 \mathcal{L}   &=\sum_{n\in\mathcal{I}} \sum_{a\in\mathcal{C}} \Big(c_{na}( \sum_{b,b'\in \mathcal{N}}\sum_{l \in \{r\xrightarrow{}b \}}\sum_{h \in \{b'\xrightarrow{}a \}}x_{l}x_{bb'}x_{h} \ln{{P}_n\{a|\{l,h\}\}}) \Big)
    \end{align*}
    Then,
    \begin{align}
   \frac{\partial \mathcal{L}}{\partial x_{bb'}} &= \sum_{n\in\mathcal{I}} \sum_{a\in\mathcal{C}} \Big(c_{na}( 
   \sum_{l \in \{r\xrightarrow{}b\}}\sum_{h \in \{b'\xrightarrow{}a\}} x_{l}x_{h}\ln{{P}_n\{a|\{l,h\}\}}+\sum_{k \in \{r\xrightarrow{}a\}}x_{k} \frac{\partial\ln{{P}_n\{a|k\}}}{\partial x_{bb'}}\big)  \Big) \tag{A.4}
    \end{align}  
To evaluate (A.4) efficiently at tree solutions, note that contribution of the term $x_{l}x_{h}\ln{{P}_n\{a|\{l,h\}\}}$ is zero unless there is a path from the root to $b$ and from $b'$ to $a$ which can be easily checked. If there is such paths at the tree solution at which the derivative is being evaluated, then $\ln{{P}_n\{a|\{l,h\}\}}$ can be computed using steps 2 to 5 of Algorithm 2. Evaluating $\sum_{k \in \{r\xrightarrow{}a\}}x_{k} \frac{\partial\ln{{P}_n\{a|k\}}}{\partial x_{bb'}}$ is somewhat more involved. For a path $k \in \{r \xrightarrow{}a\}$, let  $b^{(1)}_{k},...,b^{(s)}_{k}$ be the set of nodes visited along $k$ where $b^{(1)}_{k}=r$ and $b^{(s)}_{k}=a$, where $s$ is the length of the path $k$. Using (3.9) we have
\begin{align}
  \frac{\partial\ln{{P}_n\{a|k\}}}{\partial x_{bb'}}&= \sum_{i=2}^{s-2}(\mu_{b^{(s-i)}_{k}}-\mu_{b^{(s-i+1)}_{k}})
  \frac{\partial\Gamma_n\{b^{(s-i+1)}_{k}\}}{\partial x_{bb'}}+(\mu_{r}-\mu_{b^{(2)}_{k}}) \frac{\partial\Gamma_n\{n{b^{(2)}_{k}}\}}{\partial x_{bb'}}-\mu_r\frac{\partial\Gamma_n\{r\}}{\partial x_{bb'}}  \tag{A.5}
\end{align}
The derivative of the inclusive value of a nest $g$ with respect to $x_{bb'}$ is computed recursively. Using (3.10) we have:
\begin{equation}
    \frac{\partial\Gamma_n\{g\}}{x_{bb'}} =\mu_{g}\sum_{j \in \mathcal{N}}x_{gj}\exp(\mu_g(\Gamma_n\{j\}-\Gamma_n\{g\}))  \frac{\partial\Gamma_n\{j\}}{x_{bb'}} \tag{A.6}
\end{equation}
The base case for this recursion is 
\begin{equation}
\frac{\partial\Gamma_n\{b\}}{x_{bb'}}= \frac{1}{\mu_b} \exp(\mu_b(\Gamma_n\{b'\}-\Gamma_n\{b\}))  \tag{A.7}
\end{equation}
Any call of (A.6) will terminate with a call of (A.7) without any further recursive calls. The evaluation of (A.5) can be computed efficiently using a similar idea to Algorithm 2 where the contributions are calculated according to (A.6).

(4) First order partial derivative of the likelihood function with respect to edges between nests and  alternatives $x_{ba}$:
    \begin{align*}
  \mathcal{L}  &=\sum_{n\in\mathcal{I}} \sum_{a\in\mathcal{C}} \Big(c_{na}( \sum_{b\in \mathcal{N}}\sum_{l\in\{r\xrightarrow{}b\}}x_{l}x_{ba} \ln{{P}_n\{a|\{l,b\}\}})\Big)
    \end{align*}
     Then,
    \begin{align}
       \frac{\partial \mathcal{L}}{\partial x_{ba}}&=  \sum_{n\in\mathcal{I}}  \Big(c_{na}( \sum_{l \in \{r\xrightarrow{}b\}} x_{l} \ln{{P}_n\{a|\{l,b\}\}})+ \sum_{a'\in\mathcal{C}}c_{na'} (\sum_{h \in \{r\xrightarrow{}a'\}}x_{h} \frac{\partial\ln{{P}_n\{a'|h\}}}{\partial x_{ba}}\big) \Big)  \tag{A.8}
    \end{align}
The contribution of the term $ x_{l} \ln{{P}_n\{a|\{l,b\}\}})$ in (4.26) is zero unless there is a path from the root to nest $b$. At tree solutions this contribution can be evaluated through a modification of Algorithm 2 by starting from the nest $b$ and propagating upwards to the root. Finally the contribution of the rightmost summation in (A.8) is computed using recursion using similar expressions to (A.6) and (A.7). The base case for the recursion in this case is different, and is given by 
\begin{equation}
\frac{\partial\Gamma_n\{b\}}{x_{ba}}= \frac{1}{\mu_b} \exp(\mu_b(V_{an}-\Gamma_n\{b\}))  \tag{A.9}
\end{equation}
  \subsection*{A.3. Counting the total number of possible paths from the root to the alternatives}
  
  We first rewrite the log-likelihood function more explicitly:
\begin{align}
    \mathcal{L}&=\sum_{n\in \mathcal{I}} \sum_{a \in \mathcal{C}}c_{na} \Big[x_{ra}(\mu_rV_{an}-\mu_r\Gamma_n\{r\}) +\sum_{b\in \mathcal{N}}  x_{rb}x_{ba}(\mu_bV_{an} +(\mu_r-\mu_b)\Gamma_n\{b\}-\mu_r\Gamma_n\{r\}) + \ldots \notag \\
    &+ \sum_{b_1, b_2, \ldots, b_p\in \mathcal{N}} (x_{rb_1}\prod_{i=2}^{p-1} x_{b_ib_{i+1}}x_{b_pa}) (\mu_{b_p}V_{an} +\sum_{i=1}^{p-1}(\mu_{b_{i}}-\mu_{b_{i+1}})\Gamma_n\{b_{i+1}\}+(\mu_{r}-\mu_{b_1})\Gamma_n\{b_1\}-\mu_r\Gamma_n\{r\})\Big] \tag{A.10}
\end{align}
The total number of terms in the expression (A.10) above is $|\mathcal{I}|\cdot|\mathcal{C}|\cdot\floor*{|\mathcal{C}|! \cdot e}$. This is also the total number of paths from the root node to the alternative nodes $a \in \mathcal{C}$.\\
To see why this is so, recall that $p=|\mathcal{N}|=|\mathcal{C}|-2$ and notice that the number of terms in the square brackets is equal to the number of permutations of the index $b$. which is $p(1+p+p(p-1)+\ldots+p!) =p(1+p!\sum_{k=1}^{p-1}\frac{1}{k!})$. Now $p!\sum_{k=1}^{p-1}\frac{1}{k!} < p!\sum_{k=1}^{\infty}\frac{1}{k!} =p!e$. The difference is given by $p!\sum_{k=n}^{\infty}\frac{1}{k!}=1+p!\sum_{k=n+1}^{\infty}\frac{1}{k!} <2$.
 \clearpage

  \begin{figure}
    \centering
 \includegraphics[scale=0.50]{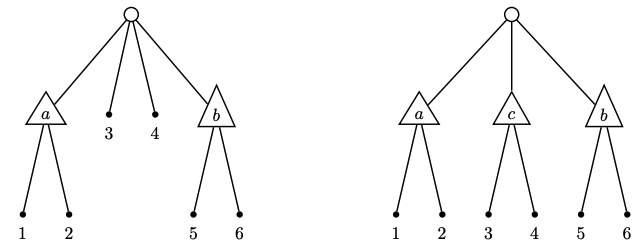} 
\caption*{\textbf{Figure A.1:} A case where it is possible to increase the number of nests while holding the nesting level fixed without worsening the likelihood}

  \end{figure}  
  
      \begin{figure}
    \centering
 \includegraphics[scale=0.5]{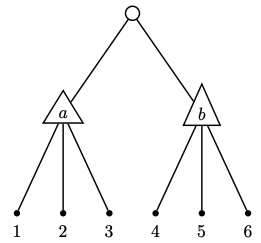} 
    \caption*{\textbf{Figure A.2:} A case where it is \textit{not} possible to increase the number of nests while holding the nesting level fixed without worsening the likelihood}
    \label{fig:my_label}
\end{figure}

     \begin{figure}
     \centering
   \includegraphics[scale=0.65]{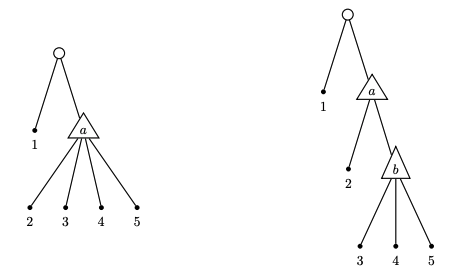} 
\caption*{\textbf{Figure A.3:} A case where increasing the number of nests and nesting level each by one does not worsen the likelihood.}
  \end{figure}

  \begin{figure}
    \centering
 \includegraphics[scale=0.55]{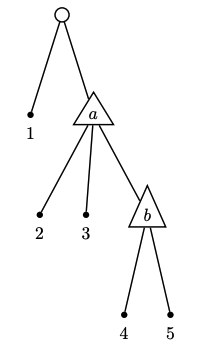} 
\caption*{\textbf{Figure A.4:} A case where it is not possible to increase the number of nests and the nesting levels each by one without worsening the likelihood.}
\end{figure}

 \newpage
\end{document}